\documentclass{article}
\usepackage[T1]{fontenc}
\usepackage{tgpagella}
\usepackage{graphicx,amsmath,amsfonts,amsthm,amscd,amssymb,bm,url,color,latexsym}
\usepackage{mathrsfs}

\renewenvironment{proof}[1][\proofname]{\noindent {\bfseries #1.}  }{\qed}
\usepackage{fullpage}
\usepackage[small,bf]{caption}
\usepackage{subcaption}
\usepackage{bbm}
\usepackage{microtype}
\usepackage{algpseudocode}
\usepackage{verbatim}
\usepackage{framed}
\usepackage{comment}
\usepackage{tikz}
\usepackage{fge}
\usepackage{mathtools}
\usepackage{enumerate}
\usepackage{multirow}
\usepackage{xcolor} 
\usepackage{soul} 
\usepackage{mathabx} 
\setstcolor{red}

\allowdisplaybreaks

\usepackage{hyperref}
\hypersetup{
    colorlinks=true,%
    citecolor=blue,%
    filecolor=blue,%
    linkcolor=blue,%
    urlcolor=blue
}
\usepackage[toc, page]{appendix}

\usepackage[capitalize,nameinlink]{cleveref}

\newtheorem{theorem}{Theorem}[section]
\newtheorem{lemma}[theorem]{Lemma}

\newtheorem{definition}[theorem]{Definition}

\newtheorem{remark}[theorem]{Remark}
\newtheorem{example}[theorem]{Example}
\newtheorem{algorithm}[theorem]{Algorithm}

\renewcommand{\mathbf}{\boldsymbol}

\newcommand{\mb}{\mathbf}
\newcommand{\mc}{\mathcal}

\newcommand{\bb}{\mathbb}

\newcommand{\set}[1]{\left\{ #1 \right\}}

\newcommand{\reals}{\bb R}

\newcommand{\eps}{\varepsilon}
\newcommand{\R}{\reals}

\newcommand{ \brac }[1]{\left[ #1 \right]}
\newcommand{ \paren }[1]{ \left( #1 \right) }

\DeclareMathOperator{\trace}{trace}

\DeclareMathOperator{\diag}{diag}

\DeclareMathOperator{\vect}{vec}

\newcommand{\wh}{\widehat}
\newcommand{\wt}{\widetilde}
\newcommand{\ol}{\overline}

\newcommand{\norm}[2]{\left\| #1 \right\|_{#2}}
\newcommand{\abs}[1]{\left| #1 \right|}
\newcommand{\innerprod}[2]{\left\langle #1,  #2 \right\rangle}
\newcommand{\prob}[1]{\bb P\left[ #1 \right]}
\newcommand{\expect}[1]{\bb E\left[ #1 \right]}

\numberwithin{equation}{section}

\title{Holographic Phase Retrieval and Reference Design}

\author{David A. Barmherzig\thanks{Institute for Computational and Mathematical Engineering, Stanford University, Stanford, CA 94305, U.S.A.}
        \and Ju Sun\thanks{Department of Mathematics, Stanford University, Stanford, CA 94305, U.S.A.}
        \and Po-Nan Li\thanks{Department of Electrical Engineering, Stanford University, Stanford, CA 94305, U.S.A.}
                \and T.J. Lane\thanks{SLAC National Accelerator Laboratory, Menlo Park, CA 94025, U.S.A.}
        \and Emmanuel J. Cand\`{e}s\thanks{Department of Mathematics and Department of Statistics, Stanford University, Stanford, CA 94305, U.S.A.}
}
\date{}

\date{  \quad Revised: \today}
\date{\today}

\begin{document}
\maketitle

\begin{abstract}
A general mathematical framework and recovery algorithm is presented for the holographic phase retrieval problem. In this problem, which arises in holographic coherent diffraction imaging, a ``reference'' portion of the signal to be recovered via phase retrieval is a priori known from experimental design. A generic formula is also derived for the expected recovery error when the measurement data is corrupted by Poisson shot noise. This facilitates an optimization perspective towards reference design and analysis. We employ this optimization perspective towards quantifying the performance of various reference choices.
\end{abstract}

\section{Introduction}

\subsection{Phase Retrieval and Coherent Diffraction Imaging}
The phase retrieval problem concerns recovering a signal from the squared magnitude of its Fourier transform. The problem can be stated symbolically as
\begin{align}  \label{eq:pr_symbol}
\begin{split}
&\textbf{Given} \quad \big{|}\wh{X}(\omega)\big{|}^2 \doteq \abs{\int_{t \in T} X(t)e^{-i \omega t}}^2\quad \text{for}\; \omega \in \Omega \\
&\textbf{Recover} \quad X
\end{split},
\end{align}
where $T$ and $\Omega$ are the (possibly multidimensional) domains of the signal and its Fourier transform, respectively. Phase retrieval arises ubiquitously in scientific imaging, where one seeks to ``image" or determine the structure of an object from various phaseless data measurements. Such settings include crystallography~\cite{PR-eg-crystallography}, diffraction imaging~\cite{PR-eg-imaging}, optics~\cite{PR-eg-optics}, and astronomy~\cite{PR-eg-astronomy}.

Phase retrieval has gained enormous attention over the last two decades, largely due to an emerging imaging technique known as Coherent Diffraction Imaging, or CDI~\cite{CDI-orig} (illustrated in~\cref{CDI}). In CDI, a coherent beam source, often being an X-ray, is illuminated upon a sample of interest. Upon the beam reaching the sample, diffraction occurs and secondary electromagnetic waves are emitted which travel until reaching a far-field detector. The detector measures the photon flux and hence records the resulting diffraction pattern, which is approximately proportional to the squared magnitude of the Fourier transform of the electric field of the sample. One can, in principle, recover the structure of the sample from the diffraction pattern by solving the phase retrieval problem~\cite{Goodman,jaganthan-thesis,eldar-review}. With the advent of extremely powerful X-ray light sources, such as X-ray Free-Electron Lasers (XFELs)~\cite{CDI-XFEL} and synchrotron radiation~\cite{CDI-synchotron}, CDI is pushing the frontier of high-resolution imaging of biological and material specimens at the nanoscale~\cite{CDI-nano-3,CDI-nano-2, CDI-nano-1,Miao530}.

\begin{figure}[!htbp] \label{CDI}
    \centering
        \includegraphics[width=0.7\textwidth]{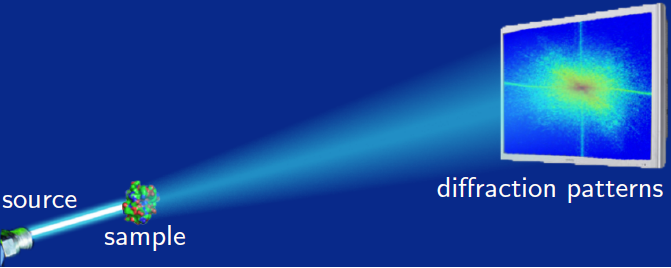}
        \caption{CDI setup. Image courtesy of~\cite{Candes-PL-Masks}.}
    \label{CDI}
\end{figure}

\subsection{Phase Retrieval Algorithms}
The phase retrieval problem does not admit a unique solution, as the forward mapping in \cref{eq:pr_symbol} maps signals related by certain intrinsic symmetries to the same set of measurements (these are discussed in detail in~\cref{sec:pr_detail}). Also, modulo these unavoidable ambiguities, there still may not be a unique solution. This nonuniqueness occurs frequently for one-dimensional signals, but only on a set of (Lebesgue) measure zero for two- or higher-dimensional signals~\cite{Hayes,eldar-review,Bendory2017,NP}. Thus, for CDI experiments (which concern two- or three-dimensional signals), there is almost surely a unique solution up to the intrinsic ambiguities. Nevertheless, solving the problem is equivalent to solving a quadratic system---which is well known to be NP-hard~\cite{Ben-TalNemirovski2001Lectures}.

In practice, the phase retrieval problem is often cast as a nonconvex optimization problem, for which various alternating-projection type algorithms are commonly employed.
The most notable one is Fienup's Hybrid Input-Output (HIO) algorithm~\cite{HIO}. Other practical variants include Relaxed Averaged Alternating Reflections (RAAR)~\cite{RAAR}, Difference Map~\cite{DifferenceMap}, and Alternating Direction Method of Multipliers (ADMM)~\cite{MarchesiniADMM,ADMM-Boyd,PR-BCD}. While often successful, these algorithms are not guaranteed to find the correct solution. They are also known to suffer from various problems such as stagnation at erroneous solutions, slow runtime, sensitivity to noise and parameter tuning~\cite{marchesini2007invited,Elser}.

To mitigate these difficulties, a line of recent work has proposed to modify the typical CDI setup. This involves sequentially modulating either the beam pattern or the Fourier transform via random or deterministic masks, thereby gathering multiple-shot measurements~\cite{PhysRevA.78.023817,PhysRevLett.100.155503,Candes-PL-Masks,deter_mask,jaganathan2015phase}. Several of these proposals have resulted in efficient algorithms with provable guarantees~\cite{Candes-PL-Masks,Candes-WF,deter_mask}. However, such a multiple-shot experiment is largely impractical, as the specimen could be damaged before the measurement process is complete~\cite{eldar-review}.

\subsection{Holographic CDI and Holographic Phase Retrieval}
In this paper, we consider another variant of CDI based on the holographic idea introduced by Gabor in 1948~\cite{Gabor1948}, which we shall term as \emph{holographic CDI}. In holographic CDI, the experiment remains single-shot, but a ``reference'' area, whose structure is a priori known, is included in the diffraction area alongside the sample of interest (see~\cref{FH-CDI} for the system setup and~\cref{fig:img-ref} for a schematic illustration).
\begin{figure}[!htbp] \label{FH-CDI}
    \centering
        \includegraphics[width=0.7\textwidth]{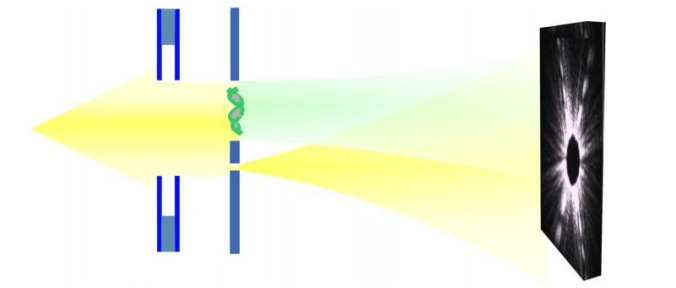}
        \caption{Holographic CDI setup. Image courtesy of~\cite{FT-Cambridge}. }
    \label{FH-CDI}
\end{figure}
\begin{figure}[!htbp] \label{fig:img-ref}
    \centering
    \includegraphics[width=0.5\textwidth]{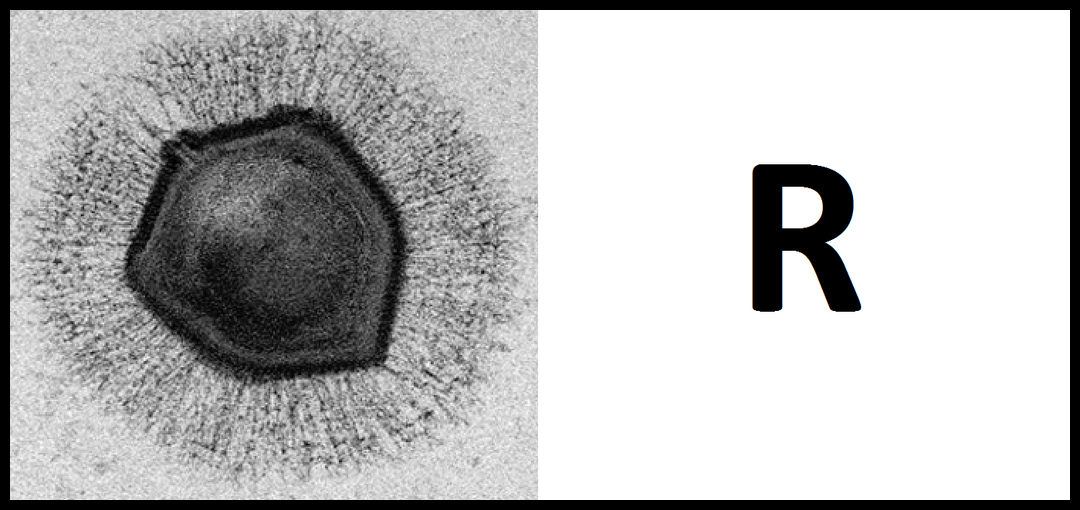}
    \caption{The diffraction area in Holographic CDI contains an unknown specimen (here shown as a mimivirus ~\cite{Mimivirus}) together with a known reference (here shown as ``$\mathbf{\mathrm{R}}$''). Popular choices for the reference $\mathbf{\mathrm{R}}$ are shown in \cref{fig:ref-compar}. }
    \label{fig:img-ref}
\end{figure}

Introducing a reference substantially simplifies the resulting phase retrieval problem, which we call \emph{holographic phase retrieval}: the computational problem is now a linear deconvolution, which is equivalent to solving a linear system\textcolor{red}~\cite{Ref-CDI-1,REF-CDI-2,HERALDO,General-Block-Ref,Kim:90}. The entailing computation can be further streamlined when certain specific reference shapes are employed. Due to its simplicity, holographic CDI is growing in its impact and popularity~\cite{Tais-annulus,FT-Cambridge}.

In the imaging community, popular reference choices are the pinhole reference~\cite{Four-holog,Kim:90}, the slit reference~\cite{HERALDO,HERALDO-2,HERALDO-real}, and the block reference~\cite{Block-ref,Block-ref-shortpaper,General-Block-Ref}, as illustrated in~\cref{fig:ref-compar}.
\begin{figure}[!htbp] \label{fig:ref-compar}
    \centering
    \begin{subfigure}[b]{0.2\textwidth}
        \includegraphics[width=\textwidth]{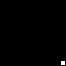}
        \caption{Pinhole reference}
        \label{holog-compar}
    \end{subfigure}
    \begin{subfigure}[b]{0.2\textwidth}
        \includegraphics[width=\textwidth]{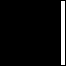}
        \caption{Slit reference}
        \label{HERALDO-compar}
    \end{subfigure}
            \begin{subfigure}[b]{0.2\textwidth}
        \includegraphics[width=\textwidth]{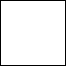}
        \caption{Block reference}
        \label{block-compar}
    \end{subfigure}
    \caption{Schematic of the three leading reference choices for Holographic CDI. (Images have $16 \times 16$ enlarged pixels for illustration.)}
    \label{fig:ref-compar}
\end{figure}
Other proposed references include L-shapes~\cite{HERALDO}, parallelograms~\cite{HERALDO}, and annuluses~\cite{Tais-annulus}, and Uniformly Redundant Arrays \cite{URA}. These reference shapes are typically realized as ``empty space'' cut out from a surrounding metal apparatus (see, e.g., \cref{FH-CDI}). For signal recovery using these references, reference-specific algorithms---which take a different approach than linear deconvolution and only apply to small classes of references---have been proposed~\cite{HERALDO}. Moreover, studies of these methods have to date been almost entirely empirical. Some error analysis is provided in~\cite{CDI-stats}.

%
%
%
%
%

\subsection{Our Contributions}
In this paper, we derive a general mathematical framework for holographic phase retrieval, encompassing the problem's setup, recovery algorithm, and error analysis.
\begin{itemize}
\item Firstly, we formulate the holographic phase retrieval problem for a general specimen and reference setup. We then provide a recovery algorithm, termed Referenced Deconvolution, which essentially amounts to solving a structured linear system. We then further show how particular reference choices simplify this linear system. This provides a novel perspective on why fast, specialized algorithms (e.g., see~\cite{Four-holog,Block-ref,HERALDO}) can be designed for these reference choices.

\item We derive a formula for the expected recovery error given noise-corrupted data. This formula offers a quantitative metric for experimental design and simulation, and allows for viewing the problem of reference design from an optimization perspective. This  formula is then specialized to the Poisson shot noise, which occurs intrinsically in CDI due to quantum mechanical principles. This leads to the key notion of the \textit{reference scaling factor}, based on which we characterize the popular references. In particular, the pinhole reference (\cref{holog-compar}) is a good choice for ``flat-spectrum'' data, whereas the block reference (\cref{block-compar}) is well suited for low-frequency dominant data.
\end{itemize}

Numerical results demonstrate the power of the proposed referenced deconvolution method and the advantage of the block reference for recovering typical CDI imaging specimens. We also view this work as a means to introduce the holographic phase retrieval problem and the optimal reference design problem to a wider mathematical and scientific audience.

\subsection{Paper Organization}
\cref{sec:alg} introduces the holographic phase retrieval problem and the referenced deconvolution algorthm. The special cases of popular reference choices are then further studied. \cref{sec:analysis} introduces an error analysis framework for holographic phase retrieval and the referenced deconvolution algorithm. This is further specialized to Poisson shot noise. The notion of a \textit{reference scaling factor} is introduced, and is shown to play a key role in analyzing the expected Poisson noise error resulting from a given reference choice. Specific error analysis is then provided for popular reference choices, and is used to compare their performance. \cref{sec:exp} presents the results of numerical simulations.

\section{Holographic Phase Retrieval and Referenced Deconvolution} \label{sec:alg}

The phase retrieval problem is introduced in  \cref{sec:pr_detail}. The holographic phase retrieval problem and the Referenced Deconvolution algorithm are then introduced in \cref{sec:holo_pr} and \cref{sec:ref_deconv}, respectively. Then, the algorithm and resulting linear system are specialized to the three popular reference choices in \cref{subsec:special-cases}.

\subsection{The Phase Retrieval Problem} \label{sec:pr_detail}
We consider the discrete two-dimensional phase retrieval problem. This discrete setting is manifested in practical CDI experiments, since CCD detectors can only take measurements at a finite number of pixel locations.

For a signal $X \in \mathbb{C}^{n_1 \times n_2}$, let $\wh{X}$ be the size $m_1 \times m_2$ discrete Fourier transform of $X$ given by
\begin{equation} \label{eqn:FT}
\wh{X}(k_1,k_2)=\sum_{t_1 = 0}^{n_1-1} \sum_{t_2=0}^{n_2-1} X(t_1,t_2)e^{-2\pi i(t_1k_1/m_1+t_2k_2/m_2)}, \quad k_1 \in \set{0, \dots, m_1-1}, k_2 \in \{0, \dots, m_2-1\}.
\end{equation}
Whenever $m_1 > n_1$ and $m_2 > n_2$, the Fourier transform is injective and is said to be \emph{oversampled}. The mapping can also be compactly expressed as matrix multiplication:
\begin{equation} \label{eqn:FT-mtrx-mult}
\wh{X}=F_L X F_R^T,
\end{equation}
where $F_L \in \bb C^{m_1 \times n_2}$ and $F_R \in \bb C^{m_2 \times n_2}$ are the corresponding discrete Fourier transform (DFT) matrices given by
\begin{align*}
F_L(k,t)=e^{-2 \pi i k t/m_1}    &\quad \text{for}\; \paren{k, t} \in \{0,\dots, m_1-1\} \times \{0, \dots, n_1-1\}, \\
F_R(k,t)=e^{-2 \pi i k t/m_2}    &\quad \text{for}\; \paren{k, t} \in  \{0,\dots, m_2-1\} \times \{0, \dots, n_2-1\}.
\end{align*}
When $m_1 \geq n_2$ and $m_2 \ge n_2$, both $F_L$ and $F_R$ have mutually orthogonal columns, and so the inverse mapping is given simply by
\begin{equation} \label{eqn:inv-FT-mtrx-mult}
X=\frac{1}{m_1 m_2}F_L^*\wh{X}(F_R^*)^T.
\end{equation}
The forward and inverse transforms in \cref{eqn:FT-mtrx-mult,eqn:inv-FT-mtrx-mult} can also be conveniently expressed as a single matrix multiplication based on matrix Kronecker products. Let $\vect(\cdot)$ be the columnwise vectorization operator acting on matrices and $\otimes$ denote the matrix Kronecker products. The following result can be directly verified:
\begin{lemma} \label{lem:mtrx-kron-prod}
$Y = AXB \Longleftrightarrow \vect(Y)=(B^T \otimes A)\vect(X)$.
\end{lemma}
Applying \cref{lem:mtrx-kron-prod} to \cref{eqn:FT-mtrx-mult,eqn:inv-FT-mtrx-mult} gives
\begin{align} \label{eqn:FT-kron}
\vect(\wh{X})& =(F_R \otimes F_L)\vect(X), \\
\vect(X) & =\frac{1}{m^2}(F_R^* \otimes F_L^*)\vect(\wh{X}), \quad \text{when}\; m_1 \ge n_1\; \text{and}\; m_2 \ge n_2.
\end{align}

We are now ready to define the phase retrieval problem in the two-dimensional, discrete setting.
\begin{definition}
The (Fourier) \textbf{phase retrieval} problem consists of recovering a signal $X \in \mathbb{C}^{n_1 \times n_2}$ given the squared magnitudes of its Fourier transform values, i.e. given the set of values $|\wh{X}(k_1,k_2)|^2, k_1 \in \{0,\hdots, m_1-1\}, k_2 \in \{0,\hdots, m_2-1\}$, which shall be denoted as $|\wh{X}(k_1,k_2)|^2$.
\end{definition}
Here, exact recovery is not possible, as the mapping $X \mapsto |\wh{X}|^2$ is not injective due to the following intrinsic ambiguities:
\begin{enumerate}
\item Global phase shift: if $X \mapsto |\wh{X}|^2$, then $e^{i \theta}X \mapsto |\wh{X}|^2$ for any $\theta \in [0, 2\pi)$;
\item Conjugate-flipping: if $X \mapsto |\wh{X}|^2$, then $X' \mapsto |\wh{X}|^2$ for $X' \in \bb C^{n_1 \times n_2}$ with $ \ol{X(n_1-1-t_1, n_2-1-t_2)} = X'(t_1, t_2)$ for all $t_1, t_2$.
\end{enumerate}
As well, circular shifts of $X$ also produce the same set of measurements if no nonzero entries are shifted past the signal domain boundaries. This gives a third intrinsic ambiguity for signals $X$ which have zero rows or columns at their boundaries. Taking all possible compositions of these operations forms a set (in fact, an equivalence class) of signals which are ``physically equivalent'' to $X$ and have exactly the same magnitude measurements $|\wh{X}|^2$. Thus, recovering $X$ from $|\wh{X}|^2$ shall be understood as recovery up to these symmetries.

\begin{definition}  \label{def:cross_corr}
For signals $X_1,X_2 \in \mathbb{C}^{n_1 \times n_2}$ both indexed over $\set{0, 1, \dots, n_1-1} \times \set{0, 1, \dots, n_2-1}$, the \emph{cross-correlation}  $C_{[X_1,X_2]} \in \mathbb{C}^{(2n_1-1)\times(2n_2-1)}$ between $X_1$ and $X_2$ is defined as
\begin{equation} \label{eqn:autocorr}
C_{[X_1,X_2]}(s_1,s_2) = \sum_{t_1 = 0}^{n_1-1}
\sum_{t_2=0}^{n_2-1} X_1(t_1,t_2)\overline{X_2(t_1-s_1,t_2-s_2)},
\end{equation}
for $\paren{s_1,s_2} \in \{-(n_1-1),\dots,n_1-1\} \times \set{-\paren{n_2-1}, \dots, n_2-1}$, where any $X_1(t_1, t_2)$ or $X_2(t_1-s_1, t_2 - s_2)$ in the summands that are outside the valid index range are taken as $0$.
\end{definition}
When $X_1$ and $X_2$ are both equal to the same signal $X$, $C_{[X_1,X_2]}$ is known as the \textit{autocorrelation} of $X$, and is denoted by $A_X$. Let $\wh{X}$ be the size $m_1 \times m_2$ Fourier transform of $X$. It is well-known that~\cite{Oppenheim}
\begin{align}
|\wh{X}|^2=\wh{A_X},
\end{align}
where
\begin{align} \label{eq:center_fft}
\begin{split}
\wh{A_X} = F_{LA} A_X F_{RA}^T \quad \text{with} \; F_{LA} \in \mathbb{C}^{m_1 \times (2n_1-1)}, m_1 \ge 2n_1-1 \quad F_{RA} \in \mathbb{C}^{m_2 \times (2n_2-1)}, m_2 \ge 2n_2-1\\
 \text{and}\; F_{LA}(k,t)=e^{-2 \pi i k t/m_1}\; \forall\; (k, t) \in \{0,\dots, m_1-1\} \times \{-(n_1-1), \dots, n_1-1\}, \\
 \text{and}\; F_{RA}(k,t)=e^{-2 \pi i k t/m_2}\; \forall\; (k, t) \in \{0,\dots, m_2-1\} \times \{-(n_2-1), \dots, n_2-1\}.
 \end{split}
\end{align}
Moreover, if $m_1\ge 2n_1-1$ and $m_2 \ge 2n_2-1$, the mapping $A_X \mapsto \wh{A_X}$ is injective, and hence $A_X$ can be recovered from $\wh{A_X}$, or equivalently $|\wh{X}|^2$, by the corresponding inverse transform:
\begin{align}
A_X=\frac{1}{m_1 m_2}F_{LA}^*|\wh{X}|^2(F_{RA}^*)^T.
\end{align}
Namely, $A_X$ is uniquely determined from the $m_1 \ge 2n_1-1, m_2 \ge 2n_2-1$ uniform frequency sampling points. The entire frequency spectrum is in turn determined by taking the 2D discrete-time Fourier transform of $A_X$. Thus, any oversampling past the $m_1=2n_1-1, m_2= 2n_2-1$ threshold provides no additional information.\footnote{This is true at least when there is no noise.} This is in some sense the phase retrieval analogue of the Shannon sampling theorem~\cite{Candes-PhaseLift}.

A powerful result by Hayes~\cite{Hayes} establishes that for all two- or higher-dimensional signals, excluding a set of Lebesgue measure zero, the only transformations on $X$ which preserve $|\wh{X}|^2$ are the physically equivalent symmetries discussed above. Thus, phase retrieval is generally well-posed in two or higher dimensions.

\subsection{Holographic Phase Retrieval and Deconvolution}   \label{sec:holo_pr}
For simplicity of exposition, henceforth we focus on square $X \in \bb C^{n \times n}$. Suppose a reference $R \in \bb C^{n \times n}$ of the same size is situated on the right next to $X$, i.e., as illustrated in \cref{fig:img-ref}. This gives $[X, R] \in \bb C^{n \times 2n}$ defined by
\begin{equation*}
[X,R](t_1,t_2) =
		\begin{cases}
			X(t_1,t_2) & \quad t_2 \in \{0,\dots, n-1\} \\
			R(t_1,t_2) & \quad t_2 \in \{n,\dots, 2n-1\}.
		\end{cases}
\end{equation*}
We may assume without loss of generality that the magnitudes of the entries of $X$ and $R$ are within the interval $[0,1]$. This convention has the physical interpretation of indicating the (average) \textit{transmission coefficient} of the specimen at each pixel location. Roughly speaking, the transmission coefficient measures the fraction of incident electromagnetic radiation that is transmitted, rather than being reflected or absorbed and is a material property---\cref{fig:transm} shows the transmission coefficient of polycarbonate at different photon energy levels. We shall subsequently consider reference setups that are physically realized as shapes cut from a surrounding opaque apparatus (i.e. consisting of metal that blocks all incident radiation). For the opaque part, the transmission coefficient is $0$, whereas for the cut-out part the coefficient is $1$.

\begin{figure}[!htbp] \label{fig:transm}
    \centering
    \includegraphics[width=0.5\textwidth]{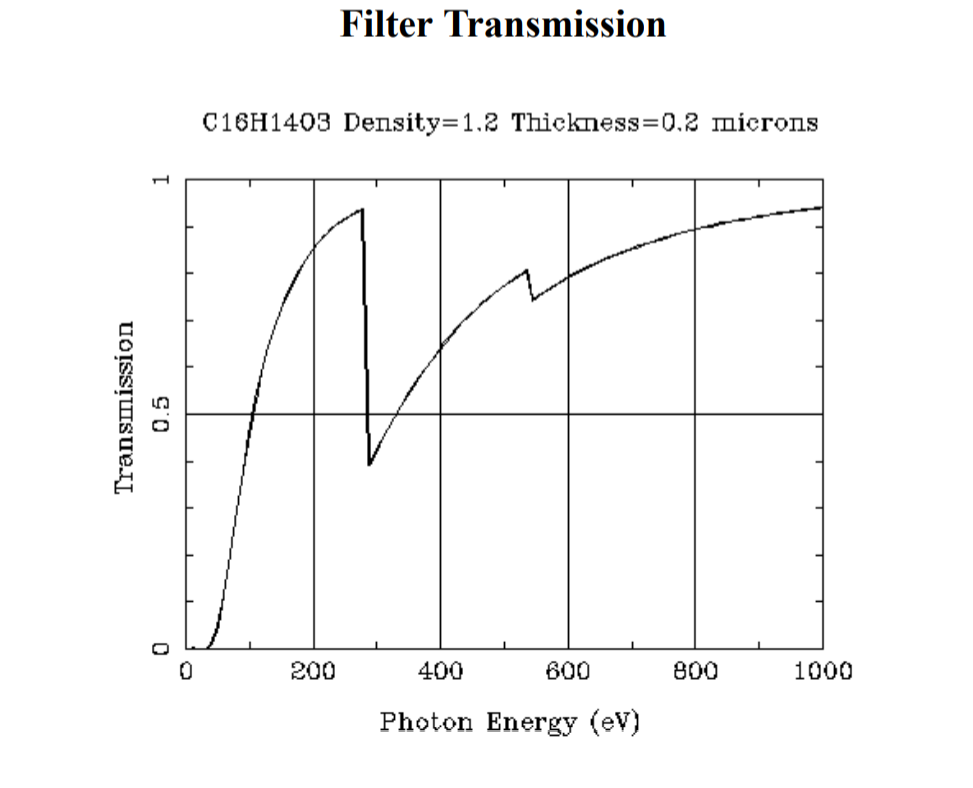}
        \caption{Transmission coefficient of polycarbonate at different photon energies ~\cite{transm}.}
        \label{fig:transm}
\end{figure}

\begin{figure}[!htbp] \label{fig:autocorr-overlap}
    \centering
    \begin{subfigure}[b]{0.28\textwidth}
        \includegraphics[width=\textwidth]{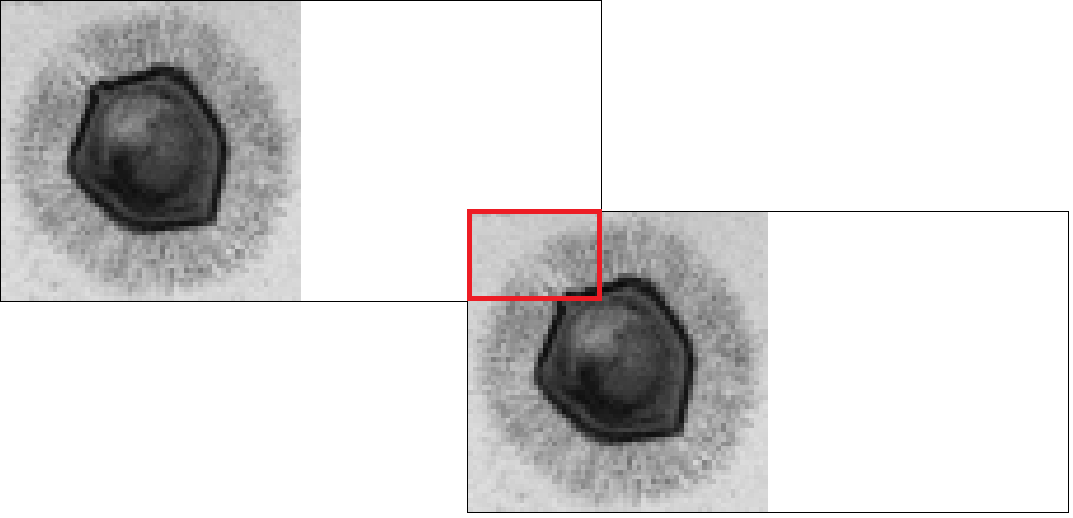}
        \caption{}
        \label{autocorr-1}
    \end{subfigure}
        \begin{subfigure}[b]{0.28\textwidth}
        \includegraphics[width=\textwidth]{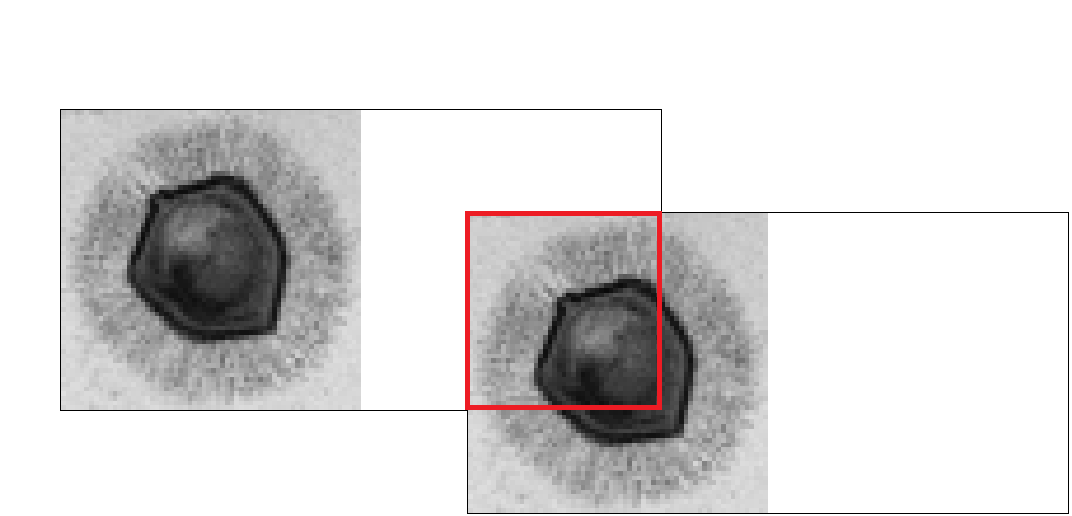}
        \caption{}
        \label{autocorr-2}
    \end{subfigure}
    \begin{subfigure}[b]{0.28\textwidth}
        \includegraphics[width=\textwidth]{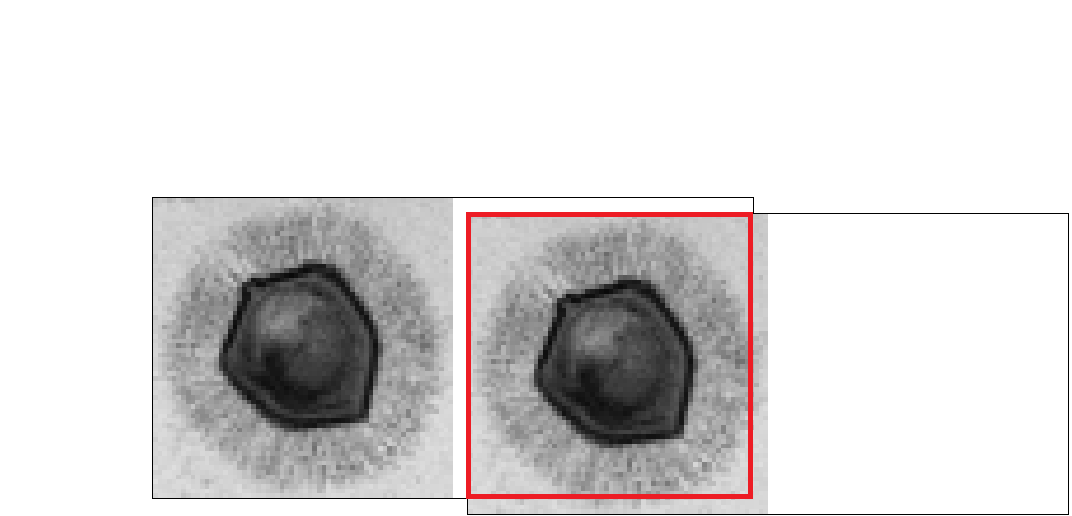}
        \caption{}
        \label{autocorr-3}
    \end{subfigure}
    \caption{Near the boundary of the index range, each successive autocorrelation value of $[X,R]$ is formed by the overlap of a subregion of $X$ with a subregion of $R$, which effectively gives a new linear measurement of $X$, as $R$ is known. The original image $X$ can then be recovered via solving the linear system in \cref{eqn:deconv}.}
    \label{fig:autocorr-overlap}
\end{figure}
In the above setting, the \emph{holographic phase retrieval} problem is:
\begin{quote}
\textbf{Holographic phase retrieval}: given $R \in \mathbb{C}^{n \times n}$ and $|\wh{\brac{X, R}}|^2 \in \R^{m_1 \times m_2}$, recover $X \in \bb C^{n \times n}$.
\end{quote}
Henceforth we assume $m_1 = m_2 = m$ with $m \ge 4n-1$. We can then recover the autocorrelation of $\brac{X, R}$, $A_{\brac{X, R}}$, by taking the inverse Fourier transform on $|\wh{\brac{X, R}}|^2$. To see how the reference $R$ helps to significantly simplify the subsequent retrieval problem, we can think over the autocorrelation process, referring to \cref{fig:autocorr-overlap}: in obtaining the autocorrelation sequence, we fix one copy of $[X, R]$, move around (in $\R^2$) another copy, and calculate and record the inner product of the overlapped region (if any) each time. When the overlap covers only part of $R$ in one copy and only part of $X$ in the other, the inner product value can be considered as a linear measurement of $X$. Since $X$ contains only $n^2$ free variables, $n^2$ non-degenerate linear measurements provide sufficient information for recovering $X$. From \cref{fig:autocorr-overlap}, we can gather $n^2$ such measurements by taking one quadrant of the cross-correlation between $X$ and $R$, which is one segment of $A_{\brac{X, R}}$! Recovering $X$ from the cross-correlation is a linear deconvolution problem.


Extracting the desired cross-correlation from $A_{[X, R]}$ is as follows. For $\paren{s_1, s_2} \in \set{-(n-1), \dots, 0} \times \set{-(n-1), \dots, 0}$,
\begin{align} \label{eq:cross_corr_XR}
C_{\brac{X, R}}\paren{s_1, s_2}
& = \sum_{t_1=0}^{n-1} \sum_{t_2 = 0}^{n-1} X\paren{t_1, t_2} \ol{R\paren{t_1-s_1, t_2 - s_2}} \nonumber \\
& = \sum_{t_1=0}^{n-1} \sum_{t_2 = 0}^{n-1} \brac{X, R}\paren{t_1, t_2} \ol{\brac{X, R}\paren{t_1-s_1, t_2 + n - s_2}} \nonumber \\
& = \sum_{t_1=0}^{n-1} \sum_{t_2 = 0}^{2n-1} \brac{X, R}\paren{t_1, t_2} \ol{\brac{X, R}\paren{t_1-s_1, t_2 + n- s_2}}  \quad (\text{``zero-filling" rule in \cref{def:cross_corr}})  \nonumber \\
& = A_{[X, R]} \paren{s_1, -n+s_2}.
\end{align}
Since this is only a quadrant of the whole cross-correlation $C_{\brac{X, R}}$, we shall write this part as $C^\diamond_{\brac{X, R}}$. The above correspondence can be compactly written as
\begin{equation} \label{eqn:bottom-right}
C^\diamond_{[X,R]}= P_1A_{[X,R]}P_2^T,
\end{equation}
where $P_1 = [I_n, 0_{n \times \paren{n-1}}]$ and $P_2 = [I_n, 0_{n \times \paren{3n-1}}]$.

For a fixed $R$, $C^\diamond_{[X,R]}$ is clearly linear in $X$. This linear relationship can be expressed conveniently as
\begin{equation} \label{eqn:deconv}
\vect(C^\diamond_{[X,R]})=M_R\vect(X),
\end{equation}
for a corresponding matrix $M_R \in \mathbb{R}^{n^2 \times n^2}$, which can be constructed by inspection of \cref{eq:cross_corr_XR}. It is easy to verify that for any choice of $R$, $M_R$ is lower-triangular and block-Toeplitz. We illustrate the form of $M_R$ using a simple example.

\begin{example} \label{eg:M-R}
Suppose
\begin{align*}
R &=
\begin{bmatrix}
    r_{00} & r_{01} & r_{02} \\
    r_{10} & r_{11} & r_{12} \\
    r_{20} & r_{21} & r_{22}
\end{bmatrix},
\end{align*}

then
\[
M_R=
\left[
\begin{array}{c|c|c}
\begin{matrix}
\overline{r_{22}} & 0 & 0 \\
    \overline{r_{12}} & \overline{r_{22}} & 0 \\
    \overline{r_{02}} & \overline{r_{12}} & \overline{r_{22}}
\end{matrix} & \begin{matrix}
0 & 0 & 0 \\
    0 & 0 & 0 \\
    0 & 0 & 0
\end{matrix} & \begin{matrix}
0 & 0 & 0 \\
    0 & 0 & 0 \\
    0 & 0 & 0
\end{matrix}\\
\hline
\begin{matrix}
\overline{r_{21}} & 0 & 0 \\
    \overline{r_{11}} & \overline{r_{21}} & 0 \\
    \overline{r_{01}} & \overline{r_{11}} & \overline{r_{22}}
\end{matrix} & \begin{matrix}
\overline{r_{22}} & 0 & 0 \\
    \overline{r_{12}} & \overline{r_{22}} & 0 \\
    \overline{r_{02}} & \overline{r_{12}} & \overline{r_{22}}
\end{matrix} & \begin{matrix}
0 & 0 & 0 \\
    0 & 0 & 0 \\
    0 & 0 & 0
\end{matrix} \\
\hline
\begin{matrix}
\overline{r_{20}} & 0 & 0 \\
    \overline{r_{10}} & \overline{r_{20}} & 0 \\
    \overline{r_{00}} & \overline{r_{10}} & \overline{r_{20}}
\end{matrix} & \begin{matrix}
\overline{r_{21}} & 0 & 0 \\
    \overline{r_{11}} & \overline{r_{21}} & 0 \\
    \overline{r_{01}} & \overline{r_{11}} & \overline{r_{21}}
\end{matrix} & \begin{matrix}
\overline{r_{22}} & 0 & 0 \\
    \overline{r_{12}} & \overline{r_{22}} & 0 \\
    \overline{r_{02}} & \overline{r_{12}} & \overline{r_{22}}
\end{matrix}
\end{array}
\right].
\]
\end{example}
Note that $M_R$ is invertible if and only if $R(n-1,n-1) \neq 0$. This invertibility condition is equivalent to the well-known ``holographic separation condition'' ~\cite{HERALDO}, dictating when an image is recoverable via a reference object. Geometrically, it guarantees that there is no aliasing corrupting the cross-correlation.

Provided that $m \ge 4n-1$ and $R(n-1,n-1) \neq 0$, we then have
\begin{align} \label{eqn:fundamental}
\begin{split}
\vect(X)
&=M_R^{-1}\vect(C^\diamond_{[X,R]})\\
&=M_R^{-1}\vect(P_1A_{[X,R]}P_2^T)\\
&=\frac{1}{m^2}M_R^{-1}\vect(P_1F_{LA}^*\big{|}\wh{[X,R]}\big{|}^2(F_{RA}^T)^*P_2^T),
\end{split}
\end{align}
where $F_{LA} \in \bb C^{m \times (2n-1)}$ and $F_{RA} \in \bb C^{m \times (4n-1)}$ are centered DFT matrices, similar to those defined in \cref{eq:center_fft}.

Applying \cref{lem:mtrx-kron-prod},
\begin{align} \label{eqn:fundamental-vec}
\vect(X) &=T_R\vect(\big{|}\wh{[X,R]}\big{|}^2),
\end{align}
where
\begin{align}  \label{eq:tr_def}
T_R=\frac{1}{m^2}M_R^{-1}\left[(P_2F_{RA}^*)\otimes(P_1F_{LA}^*)\right].
\end{align}
This gives a linear mapping between the squared Fourier transform magnitudes $|\wh{[X,R]}|^2$ and the ground truth signal $X$.

\subsection{Referenced Deconvolution}  \label{sec:ref_deconv}
Combining~\cref{eqn:fundamental,eqn:fundamental-vec} gives an algorithm for recovering $X$ given $R$ and $Y \doteq |\wh{[X,R]}|^2$. In practice, the measurements $|\wh{[X,R]}|^2$ almost always contain noise, and we shall write the possibly noisy version as $\widetilde{Y}$.
\begin{algorithm}[Referenced Deconvolution for Holographic Phase Retrieval]  \label{alg:ref_deconv}
Let $X \in \mathbb{C}^{n \times n}$ be an unknown signal, $R \in \mathbb{C}^{n \times n}$ a known ``reference'' signal with $R(n-1,n-1) \neq 0$, and $Y \doteq \wh{[X, R]}$ size $m \times m$ Fourier transform of $[X, R]$ with $m \ge 4n-1$. Let  $\widetilde{Y}$ be a noisy version of $Y$.
\begin{enumerate}
\item Given $\widetilde{Y}$, apply an inverse Fourier transform ($\bb C^{m \times m} \mapsto \bb C^{(2n-1) \times  (4n-1)}$) to obtain $\widetilde{A_{[X,R]}}$, an estimate of the autocorrelation $A_{[X,R]}$;

\item Select the top-left $n \times n$ submatrix of $\widetilde{A_{[X,R]}}$, denoted as $\widetilde{C^{\diamond}_{[X,R]}}$, which is an estimate of the top-left $n \times n$ quadrant of the cross-correlation $C_{[X, R]}$;

\item Set $\vect(\widetilde{X})= M_R^{-1}\vect(\widetilde{C^{\diamond}_{[X,R]}})$.
\end{enumerate}
\end{algorithm}

\begin{remark}
If $\widetilde{Y}=Y$, then $\widetilde{X}=X$. In other words, \cref{alg:ref_deconv} provides exact recovery in the noiseless setting.
\end{remark}

\subsection{Special Cases} \label{subsec:special-cases}

We now specialize the referenced deconvolution algorithm to three popular reference choices: the pinhole reference, the slit reference, and the (constant) block reference (see \cref{fig:ref-compar}). The different choices lead to different $M_R$'s in Step 3 of the algorithm. For all the three choices, the resulting $M_R$ can be written as a Kronecker product of two simple matrices whose inverses are also explicit. To obtain explicit expressions for $M_R^{-1}$, we shall make use of the following fact.
\begin{lemma}[\cite{trench}] \label{lem:trench}
Suppose $M_1$ and $M_2$ are invertible. If $M=M_1\otimes M_2$, then $M$ is invertible and $M^{-1}=M_1^{-1} \otimes M_2^{-1}$.
\end{lemma}
Thus, for $M_R = M_1 \otimes M_2$ with both $M_1$ and $M_2$ invertible, $M_R^{-1} = M_1^{-1} \otimes M_2^{-1}$. Invoking~\cref{lem:mtrx-kron-prod} again, we conclude that Step 3 of the above referenced deconvolution algorithm becomes
\begin{align}
\widetilde{X}  = M_2^{-1} \widetilde{C^{\diamond}_{[X,R]}} \paren{M_1^{-1}}^\top.
\end{align}

\subsubsection{Pinhole Reference}

\begin{definition} \label{def:pinhole}
The pinhole reference $R_p \in \mathbb{C}^{n \times n}$ is given by
\begin{equation} \label{eqn:Four-holog-ref}
R_p(t_1,t_2) = \left\{
        \begin{array}{ll}
            1, & \quad t_1=t_2 =n-1 \\
            0, & \quad \mathrm{otherwise}.
        \end{array}
    \right.
\end{equation}
\end{definition}

It then follows that $M_{R_p}$ is simply the identity matrix $I_{n^2}$ (i.e., $I_n \otimes I_n$) and $X$ is equal to $C_{[X,R]}^\diamond$. Our referenced deconvolution algorithm reduces to the non-iterative reconstruction procedure for Fourier holography~\cite{Four-holog}. The deconvolution procedure thus has $O(1)$ computational complexity.

\subsubsection{Slit Reference}

\begin{definition} \label{def:slit_ref}
The slit reference $R_s$ (see, e.g.,~\cref{HERALDO-compar}) is given by
\begin{equation} \label{eqn:HERALDO-ref}
R_s(t_1,t_2) = \left\{
        \begin{array}{ll}
            1, & \quad t_2=n-1 \\
            0, & \quad \mathrm{otherwise}.
        \end{array}
    \right.
\end{equation}
\end{definition}
Let $\mb 1_L \in \R^{n \times n}$ be a lower-triangular matrix consisting of all ones on and below the main diagonal. It can be verified that
\begin{align}  \label{eqn:M-R-p}
M_{R_s} = \diag\paren{\mb 1_L, \dots, \mb 1_L} = I_n \otimes \mathbf{1}_L.
\end{align}
The inverse of $\mb 1_L$ is the first-order difference matrix~\cite{Strang}:
\begin{equation} \label{eqn:D-defn}
D_n(t_1,t_2) = \left\{
        \begin{array}{ll}
            1, & \quad t_1=t_2 \\
            -1, & \quad t_1=t_2-1, \quad 1\leq t_2\leq n-1\\
            0, & \quad \text{otherwise}.
        \end{array}
    \right.
\end{equation}
Thus,
\begin{align}
\widetilde{X}  = D_n \widetilde{C^{\diamond}_{[X,R]}},
\end{align}
which only requires $O(n^2)$ operations to compute due to the sparsity structure of $D_n$.

\subsubsection{Block Reference}
\begin{definition} \label{def:block_ref}
The constant block reference $R_b$ (see, e.g., \cref{block-compar}) is given by
\begin{equation} \label{eqn:empty-space-ref}
R_b(t_1,t_2)=1, \quad t_1, t_2 \in \{0,\dots,n-1\}.
\end{equation}
\end{definition}
The corresponding $M_{R_b}$ is
\begin{equation} \label{eqn:M-R-B}
M_{R_b}=\mathbf{1}_L \otimes \mathbf{1}_L,
\end{equation}
and has inverse
\begin{equation} \label{eqn:M-R-B-inv}
M_{R_b}^{-1}=D_n \otimes D_n.
\end{equation}
Thus,
\begin{align}
\widetilde{X}  = D_n \widetilde{C^{\diamond}_{[X,R]}} D_n^\top,
\end{align}
which as well is computed in $O(n^2)$ operations due to the sparsity structure of $D_n$.

On these special references, our referenced deconvolution algorithm is equivalent to the HERALDO procedure~\cite{HERALDO}. We provide more details on the connection in~\cref{sec:connect_heraldo}.

\section{Error Analysis and Comparison} \label{sec:analysis}

Let $\widetilde{Y}$ be the measurement data subject to certain stochastic noise and $\widetilde{X}$ be the estimate of the signal of interest returned by the referenced deconvolution algorithm. Hence, $\vect(\widetilde{X})=T_R\vect(\widetilde{Y})$, where $T_R$ is as defined in~\cref{eq:tr_def}. This linear relationship allows us to derive a general formula for the expected squared recovery error in~\cref{subsec:Pois}, which is then specialized for a Poisson shot noise model. We further simplify the error formula for the three reference choices listed in~\cref{sec:analysis_special_cases}. This leads to insights regarding their recovery performance, which are discussed in~\cref{subsec:design-thry}.

\subsection{Expected Error Formula} \label{subsec:Pois}
In dealing with complex-valued matrices, we use the standard Frobenius inner product as induced by the standard complex vector inner product.
\begin{definition}
For $B,C \in \bb C^{m \times m}$, their Frobenius (or, Hilbert-Schmidt) inner product is defined as $\innerprod{B}{C} = \trace(BC^*)$. The Frobenius matrix norm is induced by this inner product in a natural way: $\norm{B}{F} = \paren{\innerprod{B}{B}}^{1/2} = \paren{\sum_{i, j} \abs{B_{i, j}}^2}^{1/2}$.
\end{definition}
The following trace-shuffling identity is also useful for our subsequent calculation.
\begin{lemma}
For complex matrices $B_1, B_2, C_1, C_2$ of compatible dimensions, $\innerprod{B_1 C_1}{B_2 C_2} = \innerprod{B_2^*B_1}{C_2C_1^*}$.
\end{lemma}
We shall use this identity in obtaining a general formula for the expected squared recovery error.
\begin{lemma}   \label{lem:exp_error}
The expected squared recovery error given by the referenced deconvolution algorithm is
\begin{equation} \label{eqn:linear-exp}
\bb E \norm{\widetilde{X} - X}{F}^2 = \innerprod{T_R^* T_R}{\bb E \brac{\vect\paren{\widetilde{Y}}-\vect\paren{Y}} \brac{\vect\paren{\widetilde{Y}} -\vect\paren{Y}}^*},
\end{equation}
where for any given reference $R$, $T_R$ is as defined in~\cref{eq:tr_def}.
\end{lemma}
\begin{proof}
Direct calculation gives
\begin{align*} \label{eqn:linear-exp}
\bb E \norm{\widetilde{X} - X}{F}^2
&= \bb E \norm{ T_R\vect \paren{\widetilde{Y}}-T_R\vect\paren{Y} }{F}^2 \\
&= \bb E \innerprod{T_R\vect\paren{\widetilde{Y}} -T_R\vect \paren{Y} }{T_R\vect\paren{\widetilde{Y}}-T_R\vect\paren{Y}} \\
&= \bb E \innerprod{T_R^* T_R}{\brac{\vect\paren{\widetilde{Y}}-\vect\paren{Y} } \brac{\vect\paren{\widetilde{Y}} -\vect\paren{Y} }^*} \\
&= \innerprod{T_R^* T_R}{\bb E \brac{\vect\paren{\widetilde{Y}}-\vect\paren{Y} } \brac{\vect\paren{\widetilde{Y}} -\vect\paren{Y} }^*},
\end{align*}
as claimed.
\end{proof}

This formula provides a very reasonable design target: given $X$ and the noise model, one can seek a reference choice $R$ which minimizes the expected squared error. This perspective forms the basis of our subsequent analysis.

We now specialize our analysis assuming a Poisson shot noise model on the data~\cite{Schottky}. Poisson shot noise occurs in any experiment in which photons are collected. It is an inherent feature of the quantum nature of photon emission, and cannot be removed by any physical apparatus~\cite{Schottky}. The model can be described as follows. Let $N_p$ be the expected (or nominal) number of photons reaching the detector. Given the squared Fourier transform magnitudes $Y=|\wh{[X,R]}|^2$, let
\begin{equation}
\norm{Y}{1} \doteq \sum_{k_1,k_2=0}^{m-1} Y(k_1,k_2).
\end{equation}
Then, the photon flux at the $\paren{k_1, k_2}$-th pixel location is given by a Poisson distribution with parameter $N_p Y(k_1,k_2)/\norm{Y}{1}$ and then scaled by $\norm{Y}{1}/N_p$. These pixel distributions are also assumed to be jointly independent~\cite{CDI-stats}. We thus have the data given by
\begin{equation} \label{eqn:data}
 \widetilde{Y} \sim \frac{\norm{Y}{1}}{N_p}\mathrm{Pois}\Big{(}\frac{N_p}{\norm{Y}{1}}Y\Big{)}.
\end{equation}

We now apply this noise model to \cref{eqn:linear-exp}. Recall that both the mean and variance of a Poisson-distributed random variable with parameter $\lambda$ are equal to $\lambda$. It then follows that \footnote{We shall denote as $\diag(\cdot)$ the operator that maps a vector to the corresponding diagonal matrix, and by $\text{Diag}(\cdot)$ the operator that maps the diagonal of a matrix to the corresponding vector.}
\begin{align*}
    \bb E \brac{\vect\paren{\widetilde{Y}}-\vect\paren{Y} } \brac{\vect\paren{\widetilde{Y}} -\vect\paren{Y} }^*
    &= \frac{\norm{Y}{1}^2}{N_p^2} \times \frac{N_p}{\norm{Y}{1}} \diag(\vect(Y)) \\
    &= \frac{\norm{Y}{1}}{N_p} \diag(\vect(Y)).
\end{align*}
Hence,
\begin{align} \label{eqn:linear-exp-fund}
\mathbb{E}\|\widetilde{X} - X\|_F^2
= \Big{\langle} T_R^*T_R, \frac{\norm{Y}{1}}{N_p}\diag(\vect(Y))\Big{\rangle}
= \frac{\norm{Y}{1}}{N_p} \innerprod{S_R}{Y},
\end{align}
where $S_R=\mathrm{reshape}\big{(}\text{Diag}(T_R^*T_R),m,m\big{)}$, and $\mathrm{reshape}(\cdot,m,m)$ is the columnwise vector-to-matrix reshaping operator. We shall term $S_R$ the \textit{reference scaling factor} corresponding to a reference $R$. Thus, the expected squared recovery error is proportional to the weighted sum of the squared frequency values in $Y$, where the weights are determined by $S_R$. $S_R$ can be efficiently computed using the following observation.

\begin{remark} \label{rem:S-R-col-norm}
For all $k_1,k_2 \in \{0,\cdots, m-1\}$, let $k = mk_1 + k_2$. Then
\begin{equation}
S_R(k_1,k_2)=\|T_R(:,k)\|_2^2,
\end{equation}
where $T_R(:,k)$ denotes the $k^{\text{th}}$ column of $T_R$.
\end{remark}

Before we compute the $S_R$'s for the special references (i.e., \cref{sec:analysis_special_cases}), here we derive a (conservative) uniform lower bound on $S_R$.
\begin{theorem} \label{thm-lower-bound}
For all reference choices $R$ (with entry magnitudes normalized within $[0,1]$), and for all $k_1,k_2 \in \{0, \dots, m-1\}$,
\begin{equation}
S_R(k_1,k_2) \geq \frac{1}{m^4}.
\end{equation}
\end{theorem}
\begin{proof}
For all $k_1, k_2 \in \set{0, \dots, m-1}$ and the corresponding $k = mk_1 + k_2$,
\begin{align*}
S_R(k_1,k_2) &= \|T_R(:,k)\|_2^2\\
&= |T_R(0,k)|^2 + \sum_{t=1}^{n^2-1}|T_R(t,k)|^2\\
&= \frac{1}{m^4}\abs{
\paren{M_R^{-1}\brac{(P_2F_{RA}^*)\otimes(P_1F_{LA}^*)}}(0,k) }^2 + \sum_{t=1}^{n^2-1}|T_R(t,k)|^2.
\end{align*}
By~\cref{lem:trench}, $M_R^{-1}$ is lower triangular. So $\paren{M_R^{-1}\brac{(P_2F_{RA}^*)\otimes(P_1F_{LA}^*)}}(0,k)$ is equal to the product of $M_R^{-1}(0, 0)$, and the first element of the $k$-th column of $(P_2F_{RA}^*)\otimes(P_1F_{LA}^*)$ which takes the form $e^{i \theta}$ for a certain $\theta$. Thus,
\begin{align}
\abs{
\paren{M_R^{-1}\brac{(P_2F_{RA}^*)\otimes(P_1F_{LA}^*)}}(0,k) }^2 = \abs{M_R^{-1}(0, 0)}^2 \ge 1,
\end{align}
where the last inequality holds, as we assume $M_R(0, 0) \in [0, 1]$. This completes the proof.
\end{proof}

\subsection{Special Cases}  \label{sec:analysis_special_cases}
For the special cases, we shall see that $T_R$ takes the form of $B_1 \otimes B_2$ for certain matrices $B_1$ and $B_2$, which motivates the following result.
\begin{lemma} \label{lem:diag-kron-fact}
Suppose $B_1 \in \bb C^{n_1 \times m_1}$, $B_2 \in \bb C^{n_2 \times m_2}$, and $B=B_1 \otimes B_2 \in \bb C^{n_1n_2 \times m_1m_2}$. Then, for $k_1 \in \{0,\dots, m_1\}$, $k_2 \in \{0,\dots, m_2\}$, and $k=m_2k_1+k_2 \in \{0, \dots, m_1m_2-1\}$, it holds that
\begin{align*}
\norm{B\paren{:, k}}{}^2 = \norm{B_1\paren{:, k_1}}{}^2
\norm{B_2\paren{:, k_2}}{}^2.
\end{align*}
\end{lemma}
\begin{proof}
By the definition of Kronecker product,
\begin{align*}
B\paren{:, k} =
\begin{bmatrix}
B_1(0, k_1) B_2(:, k_2) \\
B_1(1, k_1) B_2(:, k_2) \\
\dots \\
B_1(n_1, k_1) B_2(:, k_2)
\end{bmatrix}.
\end{align*}
Thus,
\begin{align*}
\norm{B\paren{:, k} }{}^2
= \sum_{t = 0}^{n_1} \abs{B_1(t, k_1)}^2 \norm{B_2(:, k_2)}{}^2
= \norm{B_2(:, k_2)}{}^2 \sum_{t = 0}^{n_1} \abs{B_1(t, k_1)}^2
= \norm{B_2(:, k_2)}{}^2 \norm{B_1(:, k_1)}{}^2,
\end{align*}
as claimed.
\end{proof}

Next, we make use of the result to calculate the expected recovery errors of the special references as introduced in \cref{subsec:special-cases}.

\subsubsection{Pinhole Reference}

\begin{theorem} \label{thm:holog-ref-err}
Let $R_p$ denote the pinhole reference given by \cref{def:pinhole}. For $k_1,k_2 \in \{0,\dots,m-1\}$,
\begin{align}
S_{R_p}(k_1,k_2)=\frac{n^2}{m^4}.
\end{align}
\end{theorem}
\begin{proof}
Since $M_{R_P} = I_{n^2}$, by \cref{eqn:fundamental-vec}, we have that
\begin{align*}
T_{R_p} = \frac{1}{m^2} \paren{P_2 F_{RA}^*} \otimes \paren{P_1 F_{LA}^*}.
\end{align*}
By \cref{rem:S-R-col-norm} and \cref{lem:diag-kron-fact}, for any $k_1, k_2 \in \set{0, \dots, m-1}$ and $k = mk_1 + k_2$,
\begin{align*}
S_{R_p}(k_1,k_2) = \frac{1}{m^4} \norm{\paren{P_2 F_{RA}^*} \paren{:, k_1}}{}^2 \norm{\paren{P_1 F_{LA}^*} \paren{:, k_2}}{}^2.
\end{align*}
Observing that any element in $P_2 F_{RA}^*$ or $P_1 F_{LA}^*$ has a unit norm, we conclude that
\begin{align*}
\norm{\paren{P_2 F_{RA}^*} \paren{:, k_1}}{}^2
= \norm{\paren{P_1 F_{LA}^*} \paren{:, k_2}}{}^2 = n,
\end{align*}
implying the claimed result.
\end{proof}

\subsubsection{Slit Reference}

\begin{theorem} \label{thm:slit-ref-err}
Let $R_s$ denote the slit reference given by \cref{def:slit_ref}. For $k_1,k_2 \in \{0,\dots,m-1\}$,
\begin{align*}
S_{R_s}(k_1,k_2) =\frac{n}{m^4}\brac{1 + 2\paren{n-1} \paren{1-\cos \paren{2\pi k_2 /m}}}.
\end{align*}
\end{theorem}

\begin{proof}
By the discussion below \cref{def:slit_ref},
\begin{align*}
M_{R_s}^{-1} = I_n \otimes D_n,
\end{align*}
where $D_n$ is the first-order difference matrix defined in \cref{eqn:D-defn}. So by \cref{eqn:fundamental-vec},
\begin{align*}
T_{R_s} = \frac{1}{m^2} \brac{I_n \otimes D_n} \paren{(P_2F_{RA}^*)\otimes(P_1F_{LA}^*)}
= (P_2F_{RA}^*) \otimes \paren{D_n P_1F_{LA}^*},
\end{align*}
where in the last equality we have used the ``mixed-product'' property of Kronecker products\footnote{This is the fact that $(A \otimes B)(C \otimes D)=(A \otimes C)(B \otimes D)$ for matrices $A,B,C,D$ of compatible dimensions.}.

By \cref{rem:S-R-col-norm} and \cref{lem:diag-kron-fact}, for any $k_1, k_2 \in \set{0, \dots, m-1}$ and $k = mk_1 + k_2$,
\begin{align*}
S_{R_s}(k_1,k_2) =
\frac{1}{m^4} \norm{\paren{P_2 F_{RA}^*}(:, k_1)}{}^2 \norm{\paren{D_n P_1F_{LA}^*}(:, k_2)}{}^2.
\end{align*}
Per the proof of \cref{thm:holog-ref-err}, $\norm{\paren{P_2 F_{RA}^*}(:, k_1)}{}^2 = n$. For the other term,
\begin{align*}
\norm{\paren{D_n P_1F_{LA}^*}(:, k_2)}{}^2
= & \norm{D_n P_1\brac{F_{LA}^*(:, k_2)}}{}^2\\
= & \norm{D_n
\begin{bmatrix}
e^{ -2\pi i \paren{n-1} k_2/m} \\
e^{ -2\pi i \paren{n-2} k_2/m} \\
\dots \\
 e^{ 2\pi i 0 k_2/m}
\end{bmatrix}}{}^2\\
= & 1 + \sum_{t=1}^{n-1} \abs{e^{-2\pi i t k_2/m} - e^{-2\pi i (t-1) k_2/m}}^2 \\
= & 1 + \sum_{t=1}^{n-1} \paren{2 - 2\cos \paren{2\pi k_2 /m}} \\
= & 1 + 2\paren{n-1} \paren{1 - \cos \paren{2\pi k_2 /m}},
\end{align*}
completing the proof.
\end{proof}

\subsubsection{Block Reference}

\begin{theorem} \label{thm:block-ref-err}
Let $R_b$ denote the block reference given by \cref{def:block_ref}. For $k_1,k_2 \in \{0,\dots,m-1\}$,
\begin{align}
S_{R_b}(k_1,k_2)= \frac{1}{m^4} \brac{1 + 2\paren{n-1} \paren{1-\cos \paren{2\pi k_1 /m}}} \brac{1 + 2\paren{n-1} \paren{1-\cos \paren{2\pi k_2 /m}}}.
\end{align}
\end{theorem}

\begin{proof}
As shown in \cref{eqn:M-R-B-inv},
\begin{align*}
M_{R_b}^{-1} = D_n \otimes D_n.
\end{align*}
After applying the mixed-product property of the Kronecker product and also~\cref{lem:diag-kron-fact} analogously to the above proof of the slit reference, it is clear that we only have to calculate $\norm{\paren{D_n P_1F_{LA}^*}(:, k_2)}{}^2$ and $\norm{\paren{D_n P_2F_{RA}^*}(:, k_1)}{}^2$. Performing analogous calculation as in the slit reference case then completes the proof.
\end{proof}

Note that $S_{R_b}$ achieves the uniform lower bound (i.e., $1/m^4$) around $k_1 = k_2 =0$.

\subsection{Reference Design Optimality} \label{subsec:design-thry}
For a fixed specimen $X$, we may view the recovery error given by \cref{eqn:linear-exp-fund} as an objective (i.e. cost) function whose variables are the reference values. From this perspective, each of the special cases considered exhibits a unique characteristic, as we discuss below.

\begin{figure}[!htbp] \label{ref-plot}
    \centering
        \begin{subfigure}[b]{0.31\textwidth}
        \includegraphics[width=\textwidth]{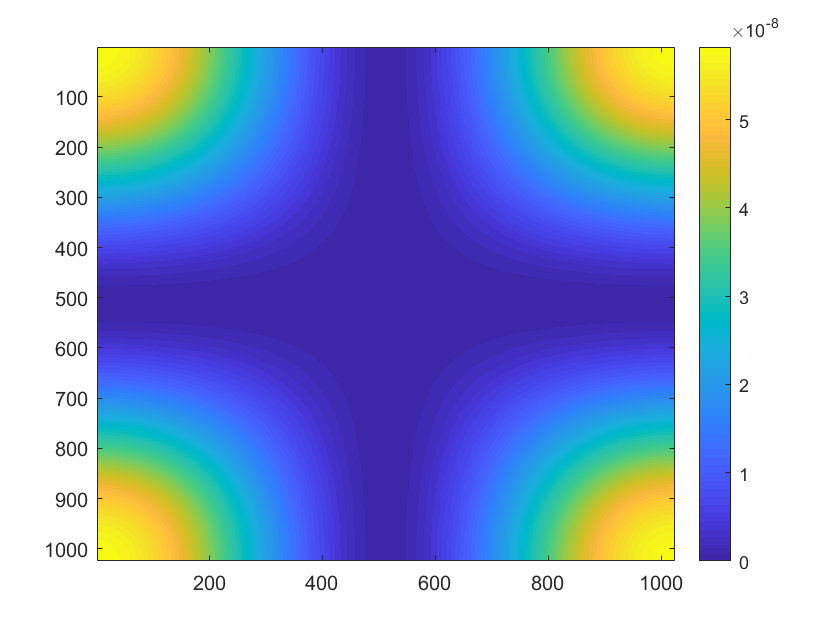}
        \caption{}
        \label{freq-Sb}
    \end{subfigure}
        \begin{subfigure}[b]{0.31\textwidth}
        \includegraphics[width=\textwidth]{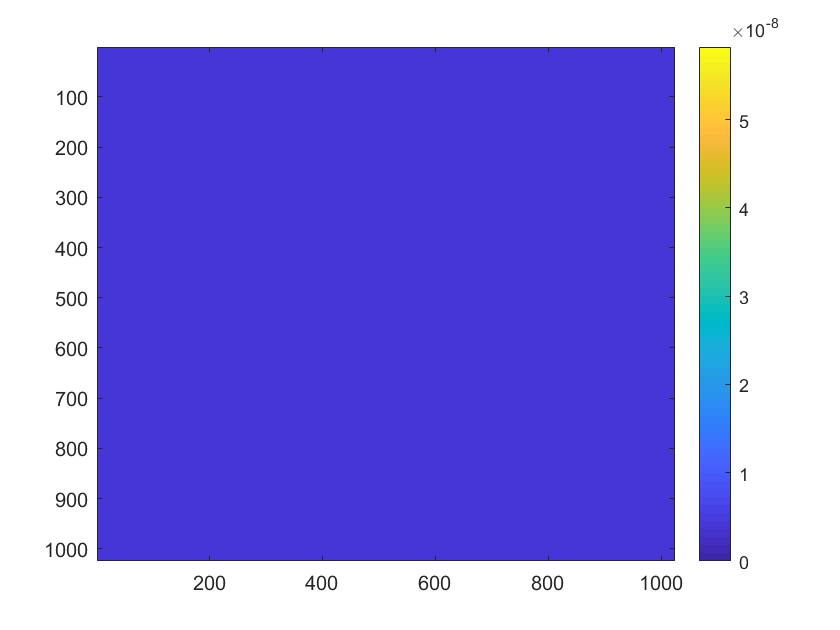}
        \caption{}
        \label{freq-Sp}
    \end{subfigure}
        \begin{subfigure}[b]{0.31\textwidth}
        \includegraphics[width=\textwidth]{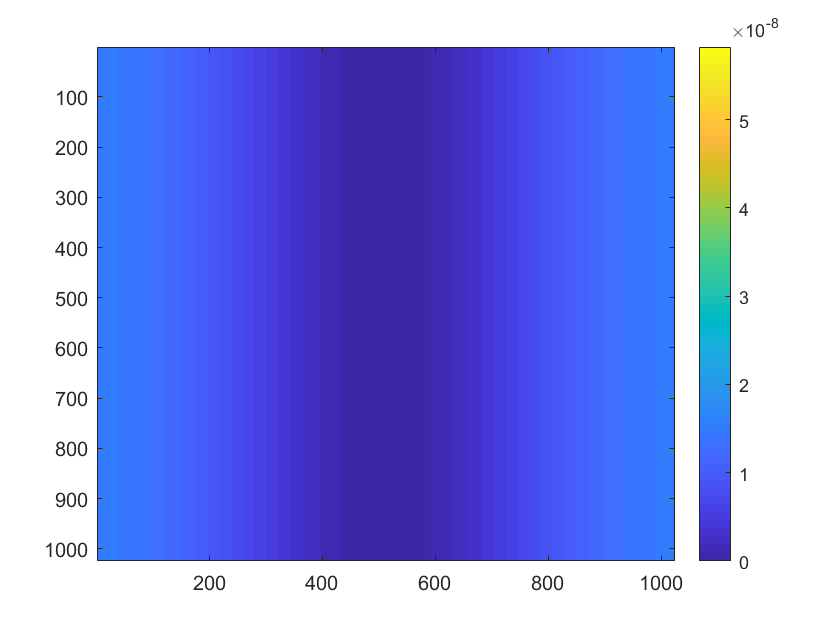}
        \caption{}
        \label{freq-Ss}
    \end{subfigure}
    \begin{subfigure}[b]{0.45\textwidth}
        \includegraphics[width=0.95\textwidth]{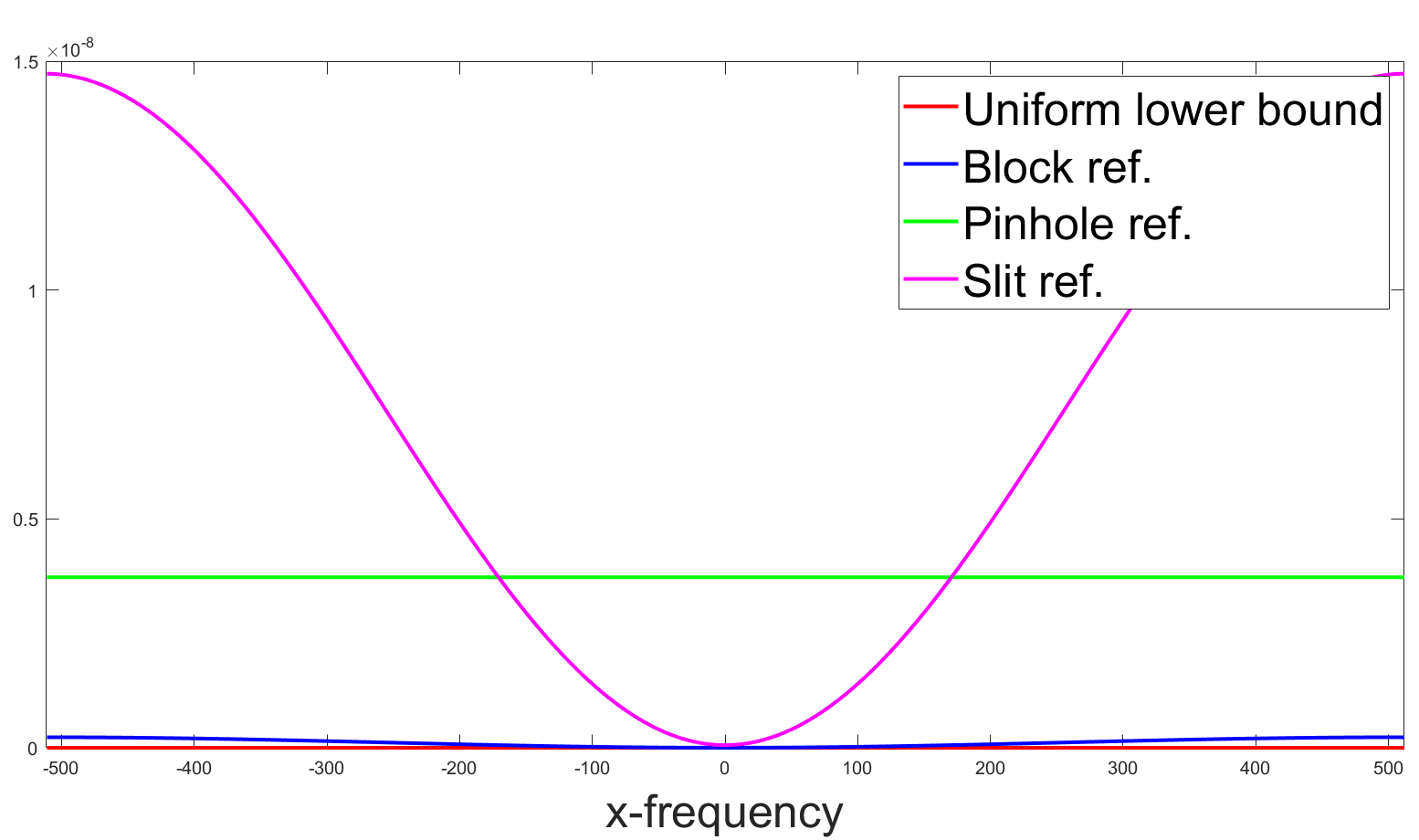}
        \caption{}
          \label{freq-scale-x}
    \end{subfigure}
            \begin{subfigure}[b]{0.45\textwidth}
        \includegraphics[width=0.95\textwidth]{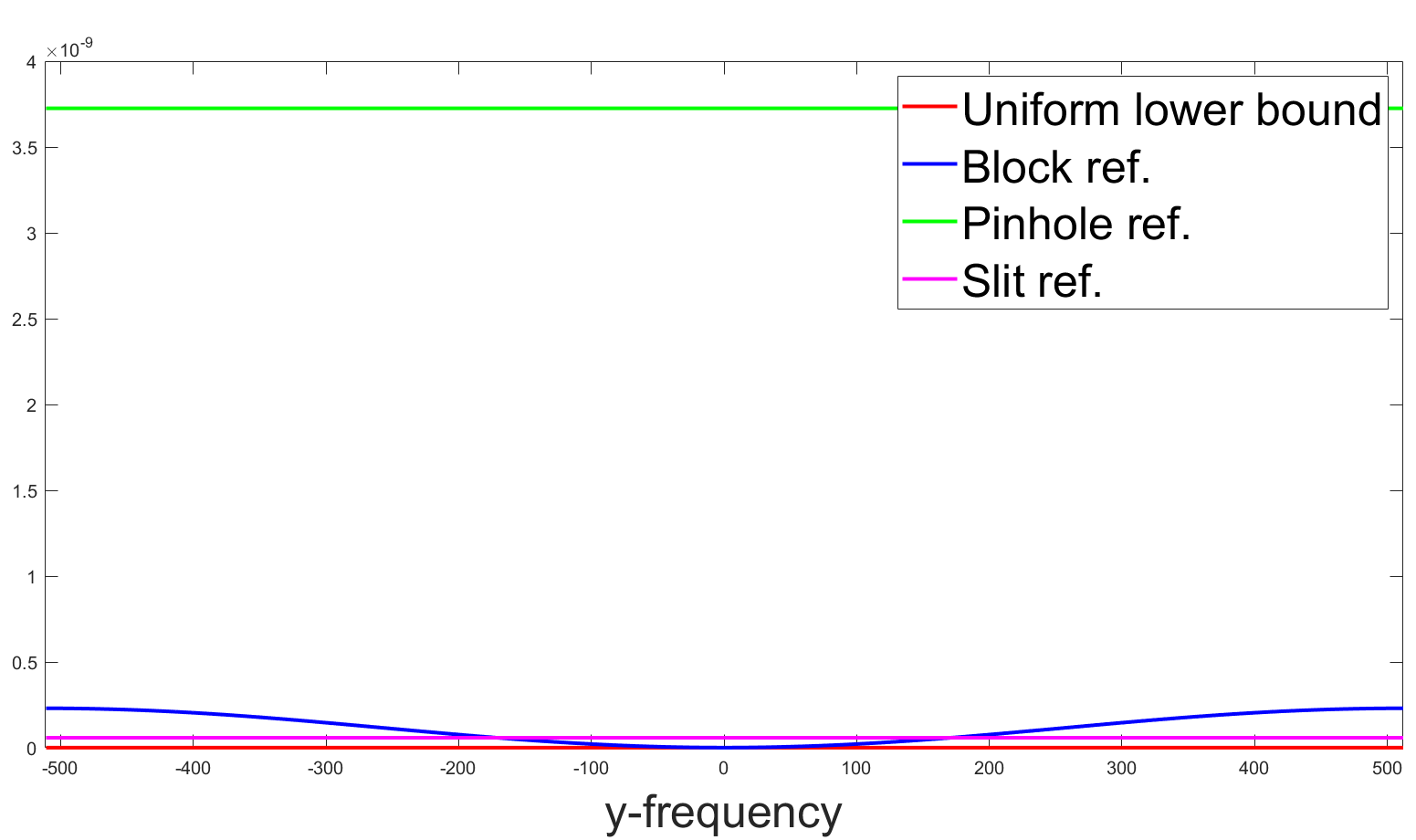}
        \caption{}
        \label{freq-scale-y}
    \end{subfigure}
    \caption{Plot of $S_R$ for the three special references and the uniform lower bound established in \cref{thm-lower-bound}.  \cref{freq-Sb},\cref{freq-Sp},\cref{freq-Ss} show the $S_R$ for the block, pinhole, and slit references, respectively. \cref{freq-scale-x} and \cref{freq-scale-y} show the cross-sections cutting the origin and parallel to the $x$ and $y$ directions, respectively. While the pinhole reference induces a flat scaling to the entire spectrum (as predicted in~\cref{thm:holog-ref-err}), the values for the block references are small at the low frequencies, and grow toward larger values moving into high frequencies. The slit reference interpolates the behaviors across directions: along the $y$-axis, the values are constant---same as the pinhole reference, but along the $x$-axis, the value increases as the frequency grows, similarly to the block reference.
    }
    \label{ref-plot}
\end{figure}

Note that these observations describe how the reference scaling factor $S_R$ depends on the reference choice $R$. This only partially describes the effect of the reference choice $R$ on the expected error given in \cref{eqn:linear-exp-fund}. Of course, the choice of $R$ will also affect $Y$ itself. However, the effect of $R$ on the former will typically far outweigh the latter by several orders of magnitude. We provide a sketch argument explaining this below.

In~\cref{eqn:linear-exp-fund}, both $Y = |\wh{[X, R]}|^2$ and $S_R$ depend on the reference $R$, and both $Y$ and $S_R$ contribute to the expected error.
\begin{itemize}
\item For $Y$, note that
\begin{align}  \label{eq:y_expand_sq}
|\wh{[X, R]}|^2
& = |\wh{[X, 0]} + \wh{[0, R]}|^2  \quad \text{(Fourier transform is a linear operator)} \\
& = |\wh{[X, 0]}|^2 + |\wh{[0, R]}|^2 + 2 \Re\paren{\wh{[X, 0]} \odot \wh{[0, R]}},
\end{align}
where $\odot$ denotes the elementwise Hadamard product. Moreover,
\begin{align}
\norm{Y}{1} = \|\wh{[X, R]}\|_F^2 = m^2 \|[X, R]\|_F^2 = m^2 \norm{X}{F}^2 + m^2\norm{R}{F}^2,
\end{align}
where the second equality follows from Parseval's theorem.

\item The $S_R$ term changes significantly across the references: for example, on a $64 \times 64$ image, the zero-frequency scaling term $S_R(0, 0)$ for the block reference is $1/64$ of that for the slit reference, and is $1/64^2$ of that for the pinhole reference, as implied by~\cref{thm:holog-ref-err,thm:slit-ref-err,thm:block-ref-err}. By~\cref{thm:holog-ref-err}, the pinhole reference induces a ``flat'' weighting scheme with a uniform weight $n^2/m^4$. By contrast, the weights induced by the block reference are frequency-varying (\cref{thm:block-ref-err}): when one of $k_1$ and $k_2$ is reasonably small, the weights are on the order $O(n/m^4)$, and when both are small, the weights are on the optimal order $O(1/m^4)$, which matches the lower bound given by~\cref{thm-lower-bound}. The weights induced by the slit reference interpolate the previous two in different directions: for a fixed $k_2$, the weight is constant and the behavior matches that of the pinhole reference, whereas the behavior is similar to that of the block reference when $k_2$ changes. The weighting behaviors of the three references are demonstrated in~\cref{ref-plot}.

\end{itemize}

To illustrate how \cref{eqn:linear-exp-fund} and the above facts can help provide insights into reference design and choice, we look at two stylized cases. For this discussion, reference choice is confined to the three special references we discussed above.
\begin{itemize}
\item \textbf{Case I: Spectrum of $X$ concentrates on (super) low-frequency bands}. A good example is when $X = \mb 1_{n \times n}$. We think of $X$ as ``flat'' and has values on the order of $\Theta(1)$. So $\norm{X}{F}^2 \in \Theta (n^2)$. Then whatever the choice of $R$, $\norm{X}{F}^2 + \norm{R}{F}^2 \in \Theta(n^2)$. So the contribution by $\norm{Y}{1}$ only differs across the three references by a small constant factor. Moreover, by \cref{eq:y_expand_sq}, $Y$ is low-frequency dominant regardless of the reference. According to our above discussion about the weight distribution of $S_R$, using the block reference might be beneficial for this class of signals (depending on of course how concentrated the low-frequency components of $Y$ is).

\item \textbf{Case II: Spectrum of $X$ is flat or has significant medium- to high-frequency components}. An idealized example is when $X = \delta(0, 0)$, and we focus on the pinhole and block references. For either of them, the medium- to high-frequency components in $S_R$ are on the same order (i.e., $O(n^2/m^4)$ from \cref{thm:holog-ref-err} and \cref{thm:block-ref-err}), and in $Y$ are on the same order also (i.e., $O(1)$\footnote{Recall we focus on the medium- to high-frequency spectrum part here. This part for $|\wh{[X, 0]}|^2$ are all ones if $X = \delta(0, 0)$. Moreover, when $R = R_p$,  both $|\wh{[0, R]} |^2$ and $2 \Re \wh{[X, 0]} \odot \wh{[0, R]}$ are $O(1)$. When $R = R_b$, this part for $|\wh{[0, R]} |^2$ are also $O(1)$ in magnitude (think of the quick decay behavior of the $\mathrm{sinc}$ function), and hence $2 \Re \wh{[X, 0]} \odot \wh{[0, R]}$ are $O(1)$ as well.   }). For the low-frequency part,  $S_{R_b}$ are $O(1/m^4)$, and $S_{R_p}$ are $O(n^2/m^4)$. Moreover, for the low-frequency part of $Y$, $|\wh{[X, R_b]}|^2$ are dominated by  $|\wh{[0, R_b]}|^2$, which is $O(n^2)$, and $|\wh{[X, R_p]}|^2$ are $O(1)$. Thus, $\innerprod{S_R}{Y}$ are on the same order whether we choose the block or the pinhole reference. So the final performance of the two is largely determined by $\norm{Y}{1}$. When $\norm{X}{F} \in o(n^2)$, say $\norm{X}{F} = 1$ (for $\delta(0, 0)$), obviously the pinhole reference is more favorable, as $\norm{R_p}{F}^2 = 1$ whereas $\norm{R_b}{F}^2 = n^2$.  But if $\norm{X}{F}^2 \in \Theta(n^2)$, we do not expect substantial differences.
\end{itemize}
It is natural to expect a smooth transition of the behaviors moving from the super-flat signal $\mb 1_{n \times n}$ to the super sharp $\delta(0, 0)$. We confirm the differential behaviors of the references empirically in~\cref{sec:exp}.

\subsection{Error Perturbation} \label{sec:pert_bound}
Our error formula \cref{eqn:linear-exp-fund} gives only the expected squared recovery error, and our results in \cref{sec:analysis_special_cases} and particularly the optimality analysis in \cref{subsec:design-thry} are derived based on this expected error. One may wonder if it is reasonable to base the analysis solely on expectation. The following perturbation bound provides some insights in this direction.

\begin{theorem} \label{thm:conc-bound}
For each of the pinhole, slit, and block references, and for the Poisson noise model on $\widetilde{Y}$ given by \cref{eqn:data}, the following holds. For all $t > 0$,
\begin{multline}
    \prob{    \abs{\norm{\wt{X} - X}{F}^2 - \bb E \norm{\wt{X} - X}{F}^2} \ge t} \\
    \le 2\exp\set{-c\min\paren{\frac{N_p m t^{1/2}}{\sqrt{3N_p \norm{Y}{1} \norm{Y}{\infty}} + 2\norm{Y}{1}} , \frac{N_p^4 m^4 t^2}{n^2 \paren{\sqrt{3N_p \norm{Y}{1} \norm{Y}{\infty}} + 2\norm{Y}{1}}^4 }}},
\end{multline}
where $c>0$ is a universal constant.
\end{theorem}

Proof of this theorem is included in \cref{sec:append-pf}. To make sense of this theorem, let us plug in some practical values for the parameters: $t = n^2 t_0$ (i.e., $t_0$ is the average pixel-wise perturbation), $N_p = C_0 m^2$ for a large constant $C_0$, $\norm{Y}{1}$ is typically $O(m^2n^2)$ and $\norm{Y}{\infty}$ is typically $O(n^2)$ (i.e., the spectrum is not too peaky). Then we have 
\begin{align}
\frac{N_p m t^{1/2}}{\sqrt{3N_p \norm{Y}{1} \norm{Y}{\infty}} + 2\norm{Y}{1}}
& = c_1 C_0\frac{m}{n} t_0^{1/2}, \\
\frac{N_p^4 m^4 t^2}{n^2 \paren{\sqrt{3N_p \norm{Y}{1} \norm{Y}{\infty}} + 2\norm{Y}{1}}^4 }
& = c_2 C_0^4 \frac{m^4}{n^6} t_0^2. 
\end{align}
Evidently, making $C_0$ large and $m$ large (relative to $n$) improves the concentration. For example, in our experiment below, $C_0 \gg 100 > n$ and $m \ge 10n$, making the above exponents reasonably large and hence the tail probability negligibly small even for small $t_0$.

\section{Numerical Simulations} \label{sec:exp}
We perform numerical experiments on two sets of data to illustrate the effectiveness of the referenced deconvolution algorithm (\cref{sec:ref_deconv}), and to corroborate the theoretical prediction on optimal reference design (\cref{subsec:design-thry}). Codes for these experiments are available at \url{https://github.com/sunju/REF_CDI}.

\subsection{On the Mimivirus Image}
\begin{figure}[!htbp] \label{fig:alg-compar}
    \centering
    \includegraphics[width=0.85\textwidth]{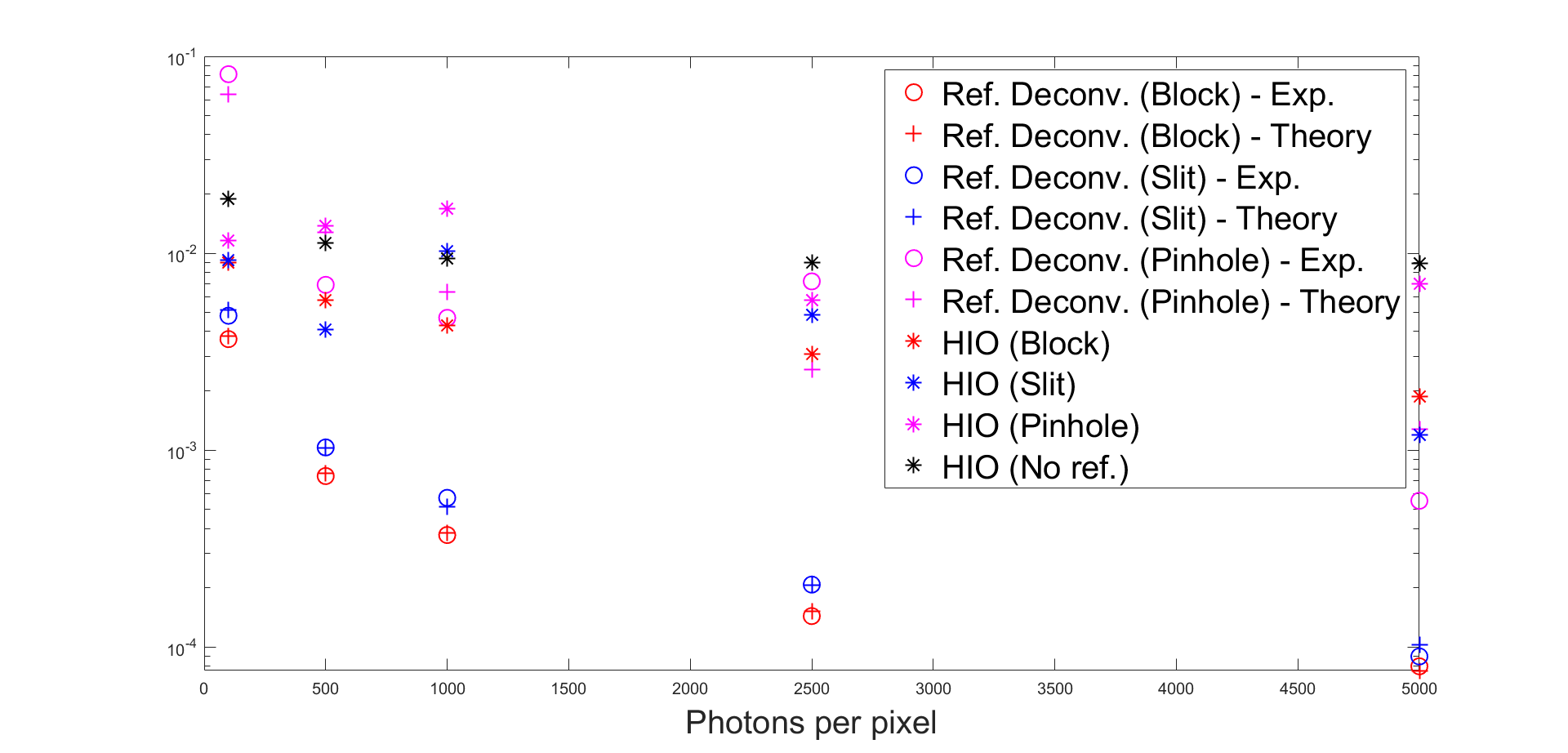}
    \caption{Plot of the relative recovery error against the nominal photon level for experiments on the mimivirus image. Result from both the referenced deconvolution and the classic HIO algorithms are included. For the referenced deconvolution algorithm, both the expected and empirical recovery errors are presented. The referenced deconvolution algorithm combined with the block reference performs consistently better than other algorithm-reference combinations. Also, for each reference, the expected recovery error closely matches the empirical recovery error, as predicted by \cref{thm:conc-bound}.}
    \label{recov}
\end{figure}

\begin{figure}[!htbp] \label{fig:mimi-recov}
    \centering
    \begin{subfigure}[b]{0.3\textwidth}
        \includegraphics[width=\textwidth]{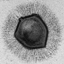}
        \caption{Ground-truth\\ image}
        \label{gnd-truth}
    \end{subfigure}
    \begin{subfigure}[b]{0.3\textwidth}
        \includegraphics[width=\textwidth]{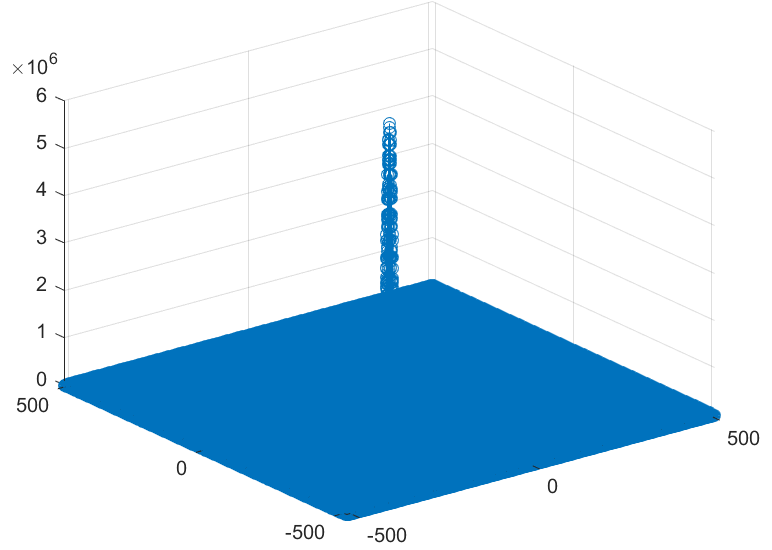}
        \caption{Fourier magnitude\\ of the groundtruth}
        \label{gnd-truth-fourier}
    \end{subfigure}
    \begin{subfigure}[b]{0.3\textwidth}
       \includegraphics[width=\textwidth]{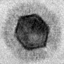}
       \caption{HIO (no ref.) \\$\eps = 93.794$, $\expect{\eps}$ NA}
       \label{HIO}
   \end{subfigure}
   \\
    \begin{subfigure}[b]{0.3\textwidth}
        \includegraphics[width=\textwidth]{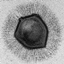}
        \caption{Ref. Deconv. with block ref. \\$\eps = 3.7029$, $\expect{\eps} = 3.7953$}
        \label{block_ref}
    \end{subfigure}
        \begin{subfigure}[b]{0.3\textwidth}
        \includegraphics[width=\textwidth]{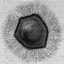}
        \caption{Ref. Deconv. with slit ref. \\$\eps = 5.7200$, $\expect{\eps} = 5.1467$}
        \label{L_ref}
    \end{subfigure}
        \begin{subfigure}[b]{0.3\textwidth}
        \includegraphics[width=\textwidth]{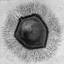}
        \caption{Ref. Deconv. with pinhole ref. \\$\eps = 46.966$, $\expect{\eps} = 63.835$}
        \label{pinhole_ref}
    \end{subfigure}
    \begin{subfigure}[b]{0.3\textwidth}
        \includegraphics[width=\textwidth]{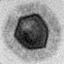}
        \caption{HIO with block ref. \\$\eps = 42.813$, $\expect{\eps}$ NA}
        \label{HIO}
    \end{subfigure}
        \begin{subfigure}[b]{0.3\textwidth}
        \includegraphics[width=\textwidth]{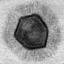}
        \caption{HIO with slit ref. \\$\eps = 102.28$, $\expect{\eps}$ NA}
        \label{HIO}
    \end{subfigure}
    \begin{subfigure}[b]{0.3\textwidth}
        \includegraphics[width=\textwidth]{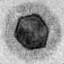}
        \caption{HIO with pinhole ref.\\$\eps = 168.18$, $\expect{\eps}$ NA}
        \label{HIO}
    \end{subfigure}
        \caption{Recovery result of the mimivirus image using various recovery schemes, and the corresponding relative recovery errors (\textbf{all errors should be rescaled by $10^{-4}$}). The photon level is fixed at $N_{pp}=1000$. Referenced deconvolution achieves superior recovery to HIO, both with and without the reference information enforced. Experimental and theoretical relative errors for referenced deconvolution closely match, as predicted by \cref{thm:conc-bound}. The block reference achieves the best recovery.}
    \label{fig:mimi-recov}
\end{figure}

In this experiment, the specimen $X$ is the mimivirus image \cite{Mimivirus}, and its spectrum mostly concentrates on very low frequencies, as shown in~\cref{gnd-truth-fourier}. The image size is $64 \times 64$, and the pixel values are normalized to $[0, 1]$. For the referenced setup, a reference $R$ of size $64 \times 64$ is placed next to $X$, forming a composite specimen $[X, R]$ of size $64 \times 128$.  Three references, i.e., the pinhole, the slit, and the block references are considered. The oversampled Fourier transform is taken to be of size $1024 \times 1024$, and  the collected noisy data $\widehat{Y}$ obeys the Poisson shot noise model defined in~\cref{eqn:data}. For this model, the nominal number of total photons $N_p$ is given by $N_{pp} \times 1024^2$, where $N_{pp}$ can be understood to be the average number of photons to be received by each pixel. We investigate the regime where $N_{pp}$ varies from $100$ to $5000$ ($N_{pp} = 100,500,1000,2500,5000$, respectively), with one simulation trial run for each $N_{pp}$ value.

We run the referenced deconvolution algorithm and also the classic HIO algorithm with and without enforcing the known reference for comparison. The results are presented in  \cref{recov} and \cref{fig:mimi-recov}. We define the relative (squared) recovery error to be
\begin{align}
\eps \doteq \frac{\norm{X - \widehat{X}}{}^2}{\norm{X}{}^2}.
\end{align}

From~\cref{recov}, it is evident that for the referenced deconvolution schemes, the expected and empirical relative recovery errors are close (justified by the perturbation result in \cref{sec:pert_bound}). Moreover, referenced deconvolution combined with the block reference performs the best among all the algorithm and reference combinations\textemdash regardless of the photon per pixel level $N_{pp}$. The superiority of the block reference among the referenced deconvolution schemes agrees with the prediction in~\cref{subsec:design-thry}, as the spectrum of $X$ sharply concentrates on very low frequencies. In addition, for the referenced deconvolution schemes, the recovery errors generally decrease as the photon level (dictated by $N_{pp}$) increases. This trend is clearly predicted by~\cref{eqn:linear-exp-fund}: because only $N_p = N_{pp} \times 1024^2$ depends on $N_{pp}$, the expected squared error is proportional to $1/N_{pp}$. The relative errors and recovered images for $N_{pp} = 1000$ are exhibited in~\cref{fig:mimi-recov}.

In regards to algorithm runtime for our experiments, the referenced deconvolution algorithm runs in less than 0.001 seconds for all three reference choices considered. The runtime for HIO is about 0.2 seconds per iteration, with the iteration with the smallest relative error selected from 1000 iterations.

\subsection{On a ``Flat-Spectrum'' Image} \label{sec:flat-spec-exp}

\begin{figure}[!htbp] \label{Flat}
    \centering
    \begin{subfigure}[b]{0.3\textwidth}
        \includegraphics[width=\textwidth]{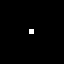}
        \caption{Ground-truth \\ image}
        \label{gnd-truth}
    \end{subfigure}
    \begin{subfigure}[b]{0.3\textwidth}
        \includegraphics[width=\textwidth]{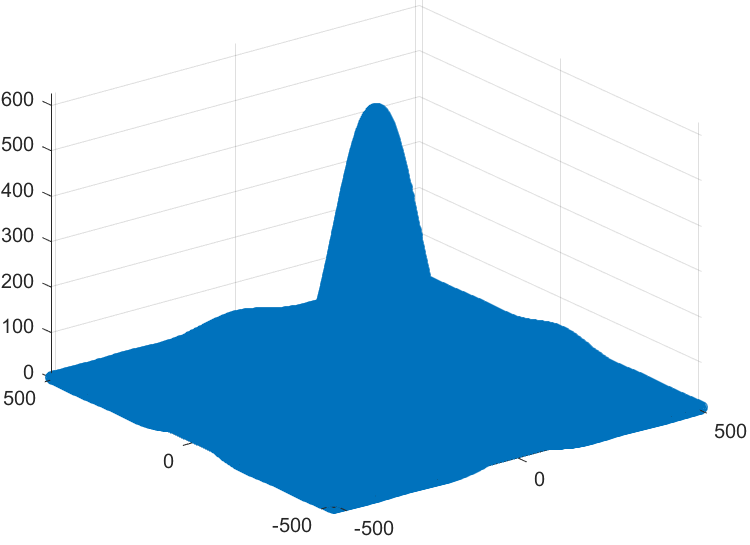}
        \caption{Fourier magnitude\\ of the groundtruth}
        \label{gnd-truth-fourier}
    \end{subfigure}
    \begin{subfigure}[b]{0.3\textwidth}
        \includegraphics[width=\textwidth]{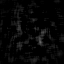}
        \caption{HIO (no ref.)\\$\eps = 21577$, $\expect{\eps}$ NA}
        \label{HIO}
    \end{subfigure}
    \\
    \begin{subfigure}[b]{0.3\textwidth}
        \includegraphics[width=\textwidth]{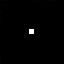}
        \caption{Ref. Deconv. with block ref. \\$\eps =  91.203$, $\expect{\eps} = 83.723$}
        \label{block_ref}
    \end{subfigure}
        \begin{subfigure}[b]{0.3\textwidth}
        \includegraphics[width=\textwidth]{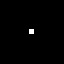}
        \caption{Ref. Deconv. with slit ref. \\ $\eps = 1.8769$, $\expect{\eps} = 1.8496$}
        \label{L_ref}
    \end{subfigure}
        \begin{subfigure}[b]{0.3\textwidth}
        \includegraphics[width=\textwidth]{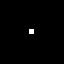}
        \caption{Ref. Deconv. with pinhole ref. \\ $\eps= 1.0609$, $\expect{\eps} = 1.0609$}
        \label{pinhole_ref}
    \end{subfigure}
    \begin{subfigure}[b]{0.3\textwidth}
        \includegraphics[width=\textwidth]{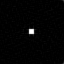}
        \caption{HIO with block ref. \\ $\eps = 760.10$, $\expect{\eps}$ NA}
        \label{HIO}
    \end{subfigure}
        \begin{subfigure}[b]{0.3\textwidth}
        \includegraphics[width=\textwidth]{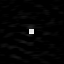}
        \caption{HIO with slit ref.\\ $\eps = 1453.1$, $\expect{\eps}$ NA}
        \label{HIO}
    \end{subfigure}
    \begin{subfigure}[b]{0.3\textwidth}
        \includegraphics[width=\textwidth]{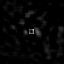}
        \caption{HIO with pinhole ref.\\ $\eps = 12855$, $\expect{\eps}$ NA}
        \label{HIO}
    \end{subfigure}

        \caption{Recovery result of the ``centered square'' image using various recovery schemes, and the corresponding relative recovery errors (\textbf{all errors should be rescaled by $10^{-4}$}). The photon level $N_{pp}$ is fixed as $1000$. Referenced deconvolution achieves superior recovery to HIO (with the reference information enforced). Experimental and theoretical relative errors for referenced deconvolution closely match, as predicted by \cref{thm:conc-bound}. The pinhole reference achieves the best recovery.}
    \label{Flat}
\end{figure}

In this experiment, the image contains a small centered square. Except for this, the basic experimental setup is identical to the above one. We focus on the case $N_{pp} = 1000$. In terms of recovery error, the referenced deconvolution schemes perform uniformly better than the HIO schemes, just as for the mimivirus image. For the current ``centered square'' image which has a considerably ``flat'' spectrum (see~\cref{gnd-truth-fourier}), however, the best-performing reference is the pinhole reference---again consistent with our theoretical prediction in~\cref{subsec:design-thry}. The detailed recovery results and recovery errors are presented in~\cref{Flat}.

\section{Discussion}
We have presented a general mathematical framework for the holographic phase retrieval problem, and proposed the referenced deconvolution algorithm as a generic solution scheme. Our formulation emphasizes the structure in the linear deconvolution procedure, and offers new insights into the resulting linear systems from popular reference choices.

We have also derived a general formula for the expected recovery error of the referenced deconvolution algorithm when the measurement data contains stochastic noise. Under a Poisson shot noise model, the formula allows us to compare popular reference choices and conclude that the block reference minimizes low-frequency contributions to the recovery error and is hence favorable for typical imaging data.

Building on our framework, it is possible to perform more detailed analysis of other noise models and reference choices. Also, the insights obtained here can likely motivate further design possibilities. In our follow-up work \cite{Dual-Ref}, one such new reference setup (termed the \textit{dual-reference}) is proposed, which provides superior noise stability across the frequency spectrum, as compared to popular single-reference setups.

Another possible extension is to include beamstops, which are often implemented in practical CDI experiments~\cite{HERALDO,Beamstop-theory}. Beamstops effectively remove a small fraction ($\le 5\%$) of low-frequency components from the measurements. We observe that the proposed referenced deconvolution algorithm is easily adapted to this setting, insofar as the missing data does not render the problem ill-conditioned.

\section{Acknowledgments}
The authors are grateful to the Simons Foundation Math+X initiative and the Natural Science and Engineering Research Council of Canada for providing support during our study. We would also like to sincerely thank Walter Murray, Gordon Wetzstein, and Jon Claerbout for ongoing valuable feedback towards developing this work.

\appendix \section{Connection with HERALDO}  \label{sec:connect_heraldo}
We sketch the correspondence between the Referenced Deconvolution algorithm and the HERALDO procedure~\cite{HERALDO}.

\subsection{Summary of HERALDO}
The HERALDO procedure (an acronym for ``Holography with Extended Reference by Autocorrelation Linear Differential Operation'') considers a continuous setup with a specimen $X(x,y)$ and reference $R(x,y)$ which together form the composite $F(x,y)=X(x,y)+R(x,y)$. Let $\star$ denote the cross-correlation.  It follows that
\begin{equation} \label{eqn:HERALDO-F-expansion}
F \star F = X \star X + R \star R + X \star R + R \star X.
\end{equation}
Provided that $R$ and $X$ satisfy certain separation conditions, the cross-correlation $R \ast X$ will not overlap with the other terms in \cref{eqn:HERALDO-F-expansion}. The HERALDO idea is to find an $n^{th}$ order weighted linear differential operator $\mathcal{L}_R^n$, i.e.,
\begin{align}
  \mc L^n_R = \sum_{k=0}^n a_k \frac{\partial^n}{\partial x^{n-k} \partial y^k}, \quad \text{where $a_k$'s are the weights}, 
\end{align}
 such that
\begin{equation} \label{eqn:HERALDO-op-defn}
\mathcal{L}^n_R(R)=A\delta(x-x_0, y-y_0) + o(x, y),
\end{equation}
for some constants $A,x_0,y_0$ and some offset function $o(x, y)$.
Now, applying the key identies 
\begin{equation} \label{eqn:HERALDO-identities}
        \mc L^{(n)} \paren{h \star g} = \mc L^{(n)} \paren{h} \star g = \paren{-1}^n h \star \mc L^{(n)} \paren{g}   \quad \forall\; h, g, 
\end{equation}
one can show that 
\begin{multline*}
\mc L^{(n)}_R \paren{F \ast F}
    = \mc L^{(n)}_R \paren{X \star X}+ \frac{1}{2} \brac{A R^*\paren{x_0 - x, y_0 - x} + o \star R} + \paren{-1}^n \frac{1}{2} \brac{A^* R\paren{x_0 + x, y_0 + x} + R \star o} +  \\
     \paren{-1}^n o \star X + X \star o + \paren{-1}^n A^* X\paren{x + x_0, y+ y_0} + A X^* \paren{x_0- x, y_0- y}.
\end{multline*}
Appropriate separation conditions on the supports (i.e., domains of nonzeros) of $X$,$R$, and $o$ then ensure either  $A^* X\paren{x + x_0, y+ y_0}$ or  $A X^* \paren{x_0- x, y_0- y}$ be separate from other terms spatially. This provides a scaled and shifted copy of $X$, thereby recovering the unknown specimen of interest.

Note that for an arbitrary reference $R$, determining whether such an $\mathcal{L}^n_R$ exists and how it can be constructed is a highly nontrivial problem. Nonetheless, for special references such as the pinhole, slit, and block, there are easy constructions as illustrated in~\cite{HERALDO}. 

\subsection{Connection to the Referenced Deconvolution Algorithm}
The HERALDO procedure seeks a continuous linear differential operator $\mathcal{L}^n_R$ such that $\mathcal{L}^n_R(R)=A\delta(x-x_0, y-y_0) + o(x, y)$. Consider the special case where $A=1$, $x_0=y_0=0$, and $o(x, y) =0$, whence 
\begin{equation} \label{eqn:HERALDO-cts-simple}
\mathcal{L}^n_R(R)=\delta(x,y).
\end{equation}
The paper~\cite{HERALDO} has derived the respective linear differential operators with $n \le 2$ for the pinhole, slit, and block references. For each of these three special references, we observe that $M_R^{-1}$ is exactly a finite-difference approximation to the linear differential operator derived in~\cite{HERALDO}\footnote{It is also possible that this correspondence can be extended to a larger class of reference choices. Indeed, it is easily seen that $M_R^{-1}\text{flip}(\vect(\overline{R}))=e_1$, whereby $M_R^{-1}$ plays the same role as does $\mathcal{L}^n_R$ in \cref{eqn:HERALDO-cts-simple}.}, which is elaborated below. 

\subsubsection{Pinhole Reference}
For the pinhole reference $R_p$, $\mathcal{L}^n_{R_p}$ is simply the identity operator, and hence its discretized version is simply the identity matrix, as is $M_{R_p}^{-1}$. 

\subsubsection{Slit Reference}
For the slit reference $R_s$, $\mathcal{L}^n_{R_s}$ is shown in Section 4.1 of~\cite{HERALDO} to be $\frac{\partial}{\partial y}$, i.e., vertical first-order differential operator, which can be approximated by $M_{R_s}^{-1}=I \otimes D_n$ (see \cref{eqn:M-R-p}), as $\paren{I \otimes D_n} \vect\paren{R_s} = \vect\paren{D_n R_s I}$. Here $D_n$ is the first-order difference matrix that performs finite differencing vertically.

\subsubsection{Block Reference}
For the slit reference $R_b$, $\mathcal{L}^n_{R_b}$ is shown in Section 4.3 of~\cite{HERALDO} to be $\frac{\partial^2}{\partial x \partial y}$, which can be approximated by $M_{R_b}^{-1}=D_n \otimes D_n$ (see \cref{eqn:M-R-B}), as $\paren{D_n \otimes D_n} \vect\paren{R_b} = \vect\paren{D_n R_b D_n^\top}$. \\

Overall, the referenced deconvolution method can be easily applied to any arbitrary reference $R$. In contrast, the HERALDO procedure can only be applied when $\mathcal{L}^n_R$ exists and can be constructed, which has to date only been demonstrated for special references with simple geometric shapes\textemdash so that the specific construction of $\mc L^n_R$ is straightforward. However, it may not be easily applicable to reference choices such as the annulus~\cite{Tais-annulus} or uniformly redundant array~\cite{URA}.

\section{Proof of \cref{thm:conc-bound}} \label{sec:append-pf}

 Let $\eta \doteq \norm{\wt{X} - X}{F}^2$. We are interested to control $\abs{\eta  - \bb E \eta}$. By similar manipulation as in \cref{lem:exp_error}, we have
 \begin{multline*}
     \norm{\wt{X} - X}{F}^2 - \bb E \norm{\wt{X} - X}{F}^2
 \\ = \innerprod{T_R^* T_R}{\brac{\vect\paren{\wt{Y}} - \vect\paren{Y}}\brac{\vect\paren{\wt{Y}} - \vect\paren{Y}}^* - \bb E \brac{\vect\paren{\wt{Y}} - \vect\paren{Y}}\brac{\vect\paren{\wt{Y}} - \vect\paren{Y}}^*}.
 \end{multline*}

Since entries in $\vect(\wt{Y}) - \vect(Y)$ have sub-exponential tails, we need Hanson-Wright type inequalities for sub-exponential random variables.
 \begin{theorem}[Proposition 1.1 of \cite{GoetzeEtAl2019Concentration}] \label{thm:subexp_chaos}
 Let $A \in \R^{n \times n}$ be a symmetric matrix and let $Z_1, \dots, Z_n$ be a set of independent random variables with $\bb E Z_i = 0$ and $\norm{Z_i}{\psi_1} \le M$ for all $i$ (here $\norm{\cdot}{\psi_1}$ is the sub-exponential norm as defined in, e.g., Definition 2.7.5. of~\cite{Vershynin2018High}). Write $Z \doteq [Z_1; \dots;Z_n]$ . For any $t> 0$,
 \begin{align}
     \prob{\abs{Z^\top A Z - \bb E Z^\top A Z }  \ge t} \le 2 \exp\paren{-c \min\set{\frac{t^2}{M^4 \norm{A}{F}^2}, \frac{t^{1/2}}{M \norm{A}{}^{1/2}}} }.
 \end{align}
 \end{theorem}
 Note that
 \begin{align*}
     \innerprod{T_R^* T_R}{\brac{\vect\paren{\wt{Y}} - \vect\paren{Y}}\brac{\vect\paren{\wt{Y}} - \vect\paren{Y}}^* }
     = \brac{\vect\paren{\wt{Y}} - \vect\paren{Y}}^* \paren{T_R^* T_R} \brac{\vect\paren{\wt{Y}} - \vect\paren{Y}}.
 \end{align*}
 So, in our problem, the random vector $Z$ is $\vect\paren{\wt{Y}} - \vect\paren{Y}$
 \begin{align}
     Z_k \doteq \wt{Y}_{ij} - Y_{ij}.
 \end{align}
 By our Poisson model,
 \begin{align}
     \expect{\wt{Y}_{ij} - Y_{ij}} = Y_{ij} - Y_{ij} = 0.
 \end{align}
 Now we need to estimate the sub-exponential norm of a centered Poisson random variable.
 \begin{lemma}  \label{lemma:subexp_norm_poisson}
     Let $Z \sim \mathrm{Pois}\paren{\lambda}$. We have
     \begin{align}
         \norm{Z -\lambda}{\psi_1} \le C \paren{\sqrt{3\lambda} + 2},
     \end{align}
     where $C > 0$ is a universal constant.
 \end{lemma}
 \begin{proof}
 Applying the Cramer-Chernoff method to $Z -\lambda$, we obtain that (see, e.g., Page 23 of~\cite{BoucheronEtAl2013Concentration})
 \begin{align}
     \prob{Z - \lambda > t} &\le \exp\set{-\lambda h\paren{t/\lambda}}  & & \quad \forall\; t \ge 0, \\
     \prob{Z - \lambda < -t} & \le \exp\set{-\lambda h\paren{-t/\lambda}}  & & \quad \forall \; t \in [0, \lambda],
 \end{align}
 where $h(u) \doteq \paren{1+u}\log (1+u) - u$. Using that $h(u) \ge \frac{u^2}{2\paren{1+u/3}}$, we can write the above results collectively as
 \begin{align}
     \prob{\abs{Z - \lambda} > t} \le 2 \exp\paren{-\frac{t^2}{2\lambda + 2t/3}} \quad \forall \; t> 0.
 \end{align}
 Now we will estimate the sub-exponential norm by upper bounding the moments $\bb E \abs{Z - \lambda}^p$. We have
 \begin{align}
     \bb E \abs{Z - \lambda}^p
     =\; & \int_{0}^{\infty} \prob{\abs{Z - \lambda}^p \ge u}\; du \\
     =\; & \int_{0}^{\infty} \prob{\abs{Z - \lambda} \ge t} p t^{p-1}\; dt \\
     \le\; & 2p \int_{0}^{\infty} \exp\paren{-\frac{t^2}{2\lambda + 2t/3}}  t^{p-1}\; dt \\
     =\;& 2p \int_{0}^{\frac{3}{2} \lambda} \exp\paren{-\frac{t^2}{2\lambda + 2t/3}}  t^{p-1}\; dt + 2p \int_{\frac{3}{2} \lambda}^{\infty} \exp\paren{-\frac{t^2}{2\lambda + 2t/3}}  t^{p-1}\; dt \\
     \le\; & 2p \int_{0}^{\frac{3}{2} \lambda} \exp\paren{-\frac{t^2}{3\lambda}}  t^{p-1}\; dt  + 2p \int_{\frac{3}{2} \lambda}^{\infty} \exp\paren{-\frac{t}{2}}  t^{p-1}\; dt \; \\
     \le\; & 2p \int_{0}^{\infty} \exp\paren{-\frac{t^2}{3\lambda}}  t^{p-1}\; dt  + 2p \int_{0}^{\infty} \exp\paren{-\frac{t}{2}}  t^{p-1}\; dt \\
     \le\; & \paren{3\lambda}^{p/2} p \Gamma \paren{p/2} + 2^{p+1} p \Gamma(p)
     \le   \paren{3\lambda}^{p/2} p \paren{p/2}^{p/2} + 2^{p+1} p^{p+1},
 \end{align}
 where we used $\Gamma(x) \le x^x$ to obtain the very last bound. Thus,
 \begin{align}
     \norm{Z - \lambda}{L_p} = \paren{\bb E \abs{Z - \lambda}^p}^{1/p} \le 5 \paren{\sqrt{3\lambda}+2} p.
 \end{align}
 We obtain the claimed result by connecting the above moment bound with the definition of sub-exponential norm, see, e.g., Proposition 2.7.1 of~\cite{Vershynin2018High}.
 \end{proof}

 We are ready now to state the concentration of the empirical (squared) error around the expectation.
 \begin{theorem}
 For any of the three special (i.e., pinhole, slit, block) references, the following holds: for all $t > 0$,
 \begin{multline}
     \prob{    \abs{\norm{\wt{X} - X}{F}^2 - \bb E \norm{\wt{X} - X}{F}^2} \ge t} \\
     \le 2\exp\paren{-c\min\set{\frac{N_p m t^{1/2}}{\sqrt{3N_p \norm{Y}{1} \norm{Y}{\infty}} + 2\norm{Y}{1}} , \frac{N_p^4 m^4 t^2}{n^2 \paren{\sqrt{3N_p \norm{Y}{1} \norm{Y}{\infty}} + 2\norm{Y}{1}}^4 }}}.
 \end{multline}
 Here $c>0$ is a universal constant.
 \end{theorem}
 \begin{proof}
 By our Poisson noise model and \cref{lemma:subexp_norm_poisson} ,
 \begin{align}
     \norm{Y_{ij} - \wt{Y}_{ij}}{\psi_1} \le C \frac{\norm{Y}{1}}{N_p} \paren{\sqrt{\frac{3N_p}{\norm{Y}{1}} Y_{ij}} + 2}.
 \end{align}
 So in applying \cref{thm:subexp_chaos}, we can take $M =  C \frac{\norm{Y}{1}}{N_p} \paren{\sqrt{\frac{3N_p}{\norm{Y}{1}} \norm{Y}{\infty}} + 2}$.

 Now we estimate $\norm{T_R^* T_R}{} = \norm{T_R}{}^2$, and $\norm{T_R^* T_R}{F}$. Note the fact that for any two matrices $A, B$, $\norm{A \otimes B}{} = \norm{A}{} \norm{B}{}$, $\norm{A \otimes B}{F} = \norm{A}{F} \norm{B}{F}$, and $\norm{AB}{F} \le \norm{A}{} \norm{B}{F}$. First, we have the following estimates
 \begin{align}
   \norm{\paren{P_2 F_{RA}^*} \otimes \paren{P_1 F_{LA}^*} }{}
   = \norm{P_2 F_{RA}^*}{} \norm{P_1 F_{LA}^*}{}
   = \sqrt{m} \times \sqrt{m} = m,
 \end{align}
 and
 \begin{align}
   \norm{\brac{\paren{P_2 F_{RA}^*} \otimes \paren{P_1 F_{LA}^*}}^* \brac{\paren{P_2 F_{RA}^*} \otimes \paren{P_1 F_{LA}^*}}}{F}
   & = \norm{\brac{\paren{P_2 F_{RA}^*} \otimes \paren{P_1 F_{LA}^*}} \brac{\paren{P_2 F_{RA}^*} \otimes \paren{P_1 F_{LA}^*}}^*}{F} \nonumber \\
   & = \norm{\brac{\paren{P_2 F_{RA}^*} \otimes \paren{P_1 F_{LA}^*}} \brac{\paren{ F_{RA} P_2^*} \otimes \paren{F_{LA} P_1^*}}}{F} \nonumber \\
   & = \norm{\paren{P_2 F_{RA}^* F_{RA} P_2^*} \otimes \paren{P_1 F_{LA}^* F_{LA} P_1^*}}{F} \nonumber \\
   & = m^2\norm{P_2 P_2^*}{F} \norm{P_1 P_1^*}{F}
   =  m^2n.
 \end{align}

 So specializing to the references, we have

 \begin{itemize}

 \item For the pinhole reference,
 \begin{align}
     \norm{T_{R_p}}{}
     & = \frac{1}{m^2} \norm{\paren{P_2 F_{RA}^*} \otimes \paren{P_1 F_{LA}^*} }{}
     = \frac{1}{m}, \\
     \norm{T_{R_p}^* T_{R_p}}{F}
     & \le \frac{1}{m^4} \norm{\brac{\paren{P_2 F_{RA}^*} \otimes \paren{P_1 F_{LA}^*}}^* \brac{\paren{P_2 F_{RA}^*} \otimes \paren{P_1 F_{LA}^*}}}{F}
     = \frac{m^2n}{m^4} = \frac{n}{m^2}.
 \end{align}

 \item For the slit reference,
 \begin{align}
     \norm{T_{R_s}}{}
     \le \frac{1}{m^2} \norm{D_n}{}  \norm{\paren{P_2 F_{RA}^*} \otimes \paren{P_1 F_{LA}^*} }{}
     \le \frac{2}{m},
 \end{align}
 where we used $\norm{D_n}{} \le 2$. Moreover,
 \begin{align}
     \norm{T_{R_s}^* T_{R_s}}{F}
     \le \frac{1}{m^4} \norm{D_n}{}^4 \norm{\brac{\paren{P_2 F_{RA}^*} \otimes \paren{P_1 F_{LA}^*}}^* \brac{\paren{P_2 F_{RA}^*} \otimes \paren{P_1 F_{LA}^*}}}{F}
     \le \frac{4n}{m^2}.
 \end{align}

 \item For the block reference,
 \begin{align}
     \norm{T_{R_b}}{}
     & \le \frac{1}{m^2} \norm{D_n}{}^2 \norm{\paren{P_2 F_{RA}^*} \otimes \paren{P_1 F_{LA}^*} }{} \le \frac{4}{m},  \\
     \norm{T_{R_b}^* T_{R_b}}{F}
     & \le \frac{1}{m^4 } \norm{D_n}{}^2  \norm{\brac{\paren{P_2 F_{RA}^*} \otimes \paren{P_1 F_{LA}^*}}^* \brac{\paren{P_2 F_{RA}^*} \otimes \paren{P_1 F_{LA}^*}}}{F}
     \le \frac{16n}{m^2}.
 \end{align}
 \end{itemize}
 So for all the three references, we have
 \begin{align}
     \norm{T_R^* T_R}{} \le \frac{16}{m^2}, \quad
     \norm{T_R^* T_R}{F} \le \frac{16n}{m^2}.
 \end{align}
 Substituting the estimates into \cref{thm:subexp_chaos}, we obtain the claimed result.
 \end{proof}


\bibliographystyle{amsalpha}
\bibliography{BREF}

\end{document}